\documentclass[a4paper,UKenglish,cleveref, autoref, thm-restate,numberwithinsect]{lipics-v2021}

\hideLIPIcs  

\usepackage{mathcomp}
\usepackage{ascmac}
\usepackage{comment}
\newtheorem*{definition*}{Definition}

\newcommand{\B}{\mathcal{B}}
\newcommand{\Ein}{E_{\text{in}}}

\title{\texorpdfstring{Reconfiguration of Labeled Matchings in Triangular Grid Graphs}{Reconfiguration of Labeled Matchings in Triangular Grid Graphs}} 

\titlerunning{Reconfiguration of labeled matchings in triangular grid graphs} 

\author{Naonori Kakimura}{Department of Mathematics, Keio University, Japan}{kakimura@math.keio.ac.jp}{https://orcid.org/0000-0002-3918-3479}{Supported by JSPS KAKENHI Grant Numbers
JP22H05001, JP20H05795, and 	23K21646, Japan and JST ERATO Grant Number JPMJER2301, Japan.}

\author{Yuta Mishima}{The Japan Research Institute, Limited, Japan}{missymissy1104@icloud.com}{}{}

\authorrunning{N. Kakimura and Y. Mishima} 

\Copyright{Naonori Kakimura and Yuta Mishima} 



\acknowledgements{The authors are grateful to Yushi Uno for valuable discussions. We are also grateful for the anonymous referees of ISAAC 2024 for their helpful comments.}

\nolinenumbers 

\EventEditors{Juli\'{a}n Mestre and Anthony Wirth}
\EventNoEds{2}
\EventLongTitle{35th International Symposium on Algorithms and Computation (ISAAC 2024)}
\EventShortTitle{ISAAC 2024}
\EventAcronym{ISAAC}
\EventYear{2024}
\EventDate{December 8--11, 2024}
\EventLocation{Sydney, Australia}
\EventLogo{}
\SeriesVolume{322}
\ArticleNo{10}

\begin{document}
\maketitle

\begin{abstract}
This paper introduces a new reconfiguration problem of matchings in a triangular grid graph.
In this problem, we are given a nearly perfect matching in which each matching edge is labeled, and aim to transform it to a target matching by sliding edges one by one.
This problem is motivated to investigate the solvability of a sliding-block puzzle called ``Gourds'' on a hexagonal grid board, introduced by Hamersma et al.~[ISAAC 2020].
The main contribution of this paper is to prove that, if a triangular grid graph is factor-critical and has a vertex of degree $6$, then any two matchings can be reconfigured to each other.
Moreover, for a triangular grid graph~(which may not have a degree-$6$ vertex), we present another sufficient condition using the local connectivity. 
Both of our results provide broad sufficient conditions for the solvability of the Gourds puzzle on a hexagonal grid board with holes, where Hamersma et al. left it as an open question.
\end{abstract}

\section{Introduction}

Combinatorial reconfiguration is a fundamental research subject that studies the solution space of combinatorial problems.
A typical example is solving sliding-block puzzles such as the 15-puzzle.
The 15-puzzle can be viewed as the transformation between the arrangement of puzzle pieces, and the goal is to transform an initial arrangement of pieces to a given target arrangement.
Combinatorial reconfiguration has applications in a variety of fields such as mathematical puzzles, operations research, and computational geometry.
See the surveys by Nishimura~\cite{N18} and van den Heuvel~\cite{Heuvel13}.

Hamersma et al.~\cite{ref:gourds} introduced a new sliding-block puzzle on a hexagonal grid, which they call \textit{Gourds}.
The name ``gourd'' refers to the shape of the puzzle pieces, which are essentially $1 \times  2$ pieces on a board.
Like in the 15-puzzle, only one grid cell is empty.
The goal is to obtain a target configuration of pieces on the board by moving pieces one-by-one, similar to the $15$-puzzle. 
Here we allow a piece to make three different kinds of moves: slide, turn, and pivot~(see~\cite{ref:gourds} for the details).
They characterized hexagonal grid boards without holes such that the Gourds puzzle\footnote{In this paper, the Gourds puzzle refers to the numbered type in~\cite{ref:gourds} where each piece has numbers.} is always solvable, and left it as a main open question to characterize boards \textit{with holes}.

Motivated by the study of the Gourds puzzle,
we introduce a reconfiguration problem of matchings in a triangular grid graph.
In the problem, we are given a matching that exposes only one vertex, which is called a \textit{nearly perfect matching}.
Each matching edge, which corresponds to a puzzle piece, is labeled.
We are allowed to slide a matching edge toward the exposed vertex.
The goal is to move matching edges one-by-one to obtain a target labeled matching.
It should be emphasized that each edge in the given matching has to be moved to the edge with the same label in the target matching. 
See Section~\ref{sec:pre} for the formal definition.
We remark that our problem can be defined on a general graph, which may be of independent interest.
The problem setting is different from the matching reconfiguration problems studied in the literature.
See Section~\ref{sec:relatedwork}.

In this paper, we investigate the  reconfigurability of the above reconfiguration problem of labeled matchings on a triangular grid graph.
In particular, we aim to characterize a triangular grid graph such that any two labeled matchings can be reconfigured to each other.
We call such a graph \textit{reconfigurable}.

As mentioned in Section~\ref{sec:factor-critical}, it is not difficult to observe that, if a graph is reconfigurable, then it is $2$-connected and factor-critical.
A graph is \textit{factor-critical} if it has a nearly perfect matching that exposes any vertex.
These two conditions, however, are not sufficient, as there exists a $2$-connected factor-critical graph that is not reconfigurable.

The main contribution of this paper is to prove that, if a $2$-connected factor-critical triangular grid graph has at least one vertex of degree $6$, then it is reconfigurable.
Our results can be adapted to the Gourds puzzle by taking the dual of a triangular grid graph, which implies that the Gourds puzzle can always be solved when at least one hexagonal cell on the board does not touch the holes or the outer face.

The key idea to prove the main result is to exploit the ear decomposition in matching theory.
A factor-critical graph is known to have a constructive characterization using ear decomposition with odd paths and cycles.
Using the ear structure, we show that, if an ear decomposition starts from a reconfigurable subgraph, then we can recursively find reconfiguration steps between any two labeled matchings.
However, every ear decomposition does not necessarily satisfy the above assumption.
We then investigate the matching structure of a triangular grid graph to identify simple reconfigurable subgraphs such that every triangular grid graph with a vertex of degree $6$ admits an ear decomposition starting from one of them.

In addition, for a triangular grid graph~(which may not have a vertex of degree $6$), we present another sufficient condition for the reconfigurability using the local connectivity.
A graph is said to be \textit{locally-connected} if the neighbor vertices of any vertex induce a connected graph.
We prove that, if a triangular grid graph is locally-connected, but not isomorphic to the Star of David graph~(Figure~\ref{fig:davide}), then it is reconfigurable.
Moreover, we show that, for a graph with $2n+1$ vertices, we can find in polynomial time reconfiguration steps with length $O(n^3)$.

The characterization for the Gourds puzzle by Hamersma et al.~\cite{ref:gourds} implies that a triangular grid graph is reconfigurable if it is $2$-connected, but not isomorphic to the Star of David graph, and has no \textit{holes}, where a hole is a face with boundary length at least $6$.
Our conditions, which allow to have a hole, are much broader than theirs, as the local connectivity and the $2$-connectivity are equivalent for a graph with no holes.


\subsection{Related Work}\label{sec:relatedwork}

A factor-critical graph plays an important role in matching theory.
It is known that any non-bipartite graph can be decomposed in terms of maximum matching, called the \textit{Gallai-Edmonds} decomposition.
It essentially decomposes a given graph into factor-critical graphs, graphs with perfect matchings, and bipartite graphs. 
Also, factor-critical graphs are used to describe facets of the matching polytope of a given graph.
See~\cite{ref:ear_decomposition,schrijver-book} for the details.

Sliding-block puzzles have been investigated in both recreational mathematics and algorithms research. 
The 15-puzzle was introduced as a prize problem by Sam Loyd in 1878~\cite{slocum200615}.
In the 15-puzzle, it is characterized by odd/even permutations whether any configuration can be realized or not~\cite{AJM1879}.
However, it is NP-complete for $n \times n$ boards to compute the smallest number of steps to reach a given configuration~\cite{DemaineR18,RatnerW86}. 
There are many variants of the $15$-puzzle such as Rush Hour~\cite{BrunnerCDHHSZ21,FlakeB02} and rolling-block puzzles~\cite{BuchinB12}. 
Many puzzles have been shown NP-hard or PSPACE-hard~(see e.g., \cite{HearnDemaine}).

In the literature of combinatorial reconfiguration, the reconfiguration of matchings has been studied extensively.
Ito et al.~\cite{ItoDHPSUU11} initiated to study a reconfiguration problem of matchings.
The aim is to decide whether a given matching can be transformed to a target matching by adding/removing a matching edge in each step.
They showed that the problem can be solved in polynomial time with the aid of the Gallai-Edmonds decomposition.
On reconfigurating perfect matchings,
Ito et al.~\cite{ItoKKKO22} studied the shortest transformation of perfect matchings by taking the symmetric difference along an alternating cycle, motivated by the study of a diameter of the matching polytope~(see also~\cite{CardinalS23,Sanita18}).
Bonamy et al.~\cite{BonamyBHIKMMW19} restricts the length of alternating cycles used in the transformations to be $4$.
We remark that all the above mentioned problems aim to transform an initial (perfect) matching to a target one in which their matching edges are \textit{not} labeled.

By taking the line graph of a given graph, 
reconfiguration problems of matchings can be viewed as reconfiguration problems of independents sets.
Our problem setting is related to its variant known as the token sliding problem.
The token sliding problem is PSPACE-complete, even on restricted graph classes such as planar graphs~\cite{HearnD05}.
On the other hand, the problem can be solved in linear time on trees~\cite{DemaineDFHIOOUY15}, and it is fixed-parameter tractable on bounded degree graphs~\cite{BartierBM23}.
See also~\cite{abs-2204-10526} for the survey on the independent set and dominating set reconfiguration problems.
Another related problem is the token swapping problem.
In the problem, we are given tokens on each vertex of a graph, and we want to move every token to its target position by swapping  two adjacent tokens.
See e.g.,~\cite{AichholzerDKLLM22,BonnetMR18, Kim16,MiltzowNORTU16} and references therein.

\begin{figure}[t]
  \begin{minipage}[b]{0.48\columnwidth}
  \centering
\includegraphics[keepaspectratio,width=0.37\textwidth]
  {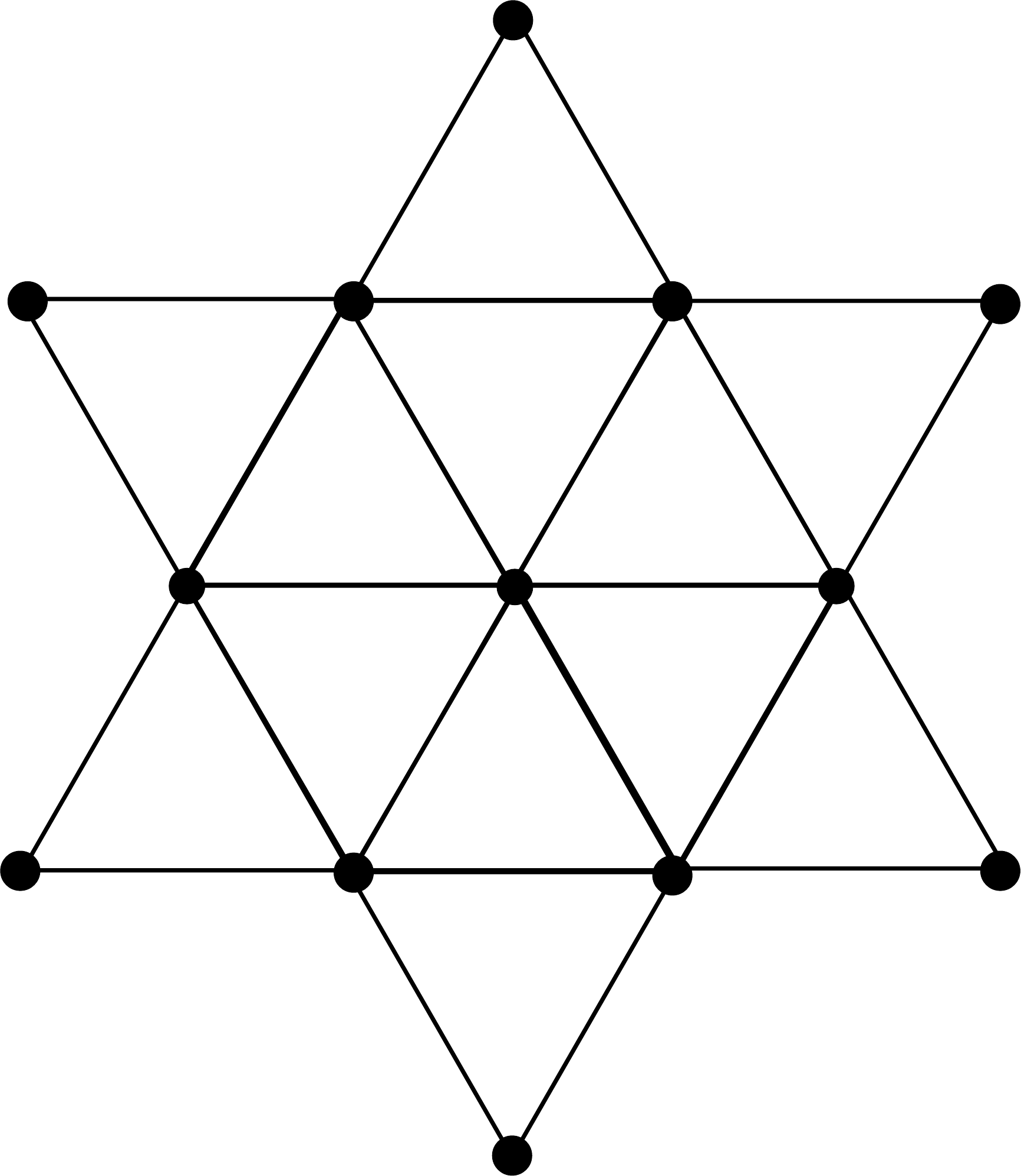}
  \caption{The Star of David graph.}
  \label{fig:davide}
  \end{minipage}
  \begin{minipage}[b]{0.48\columnwidth}
    \centering
\includegraphics[keepaspectratio,width=0.85\textwidth]
  {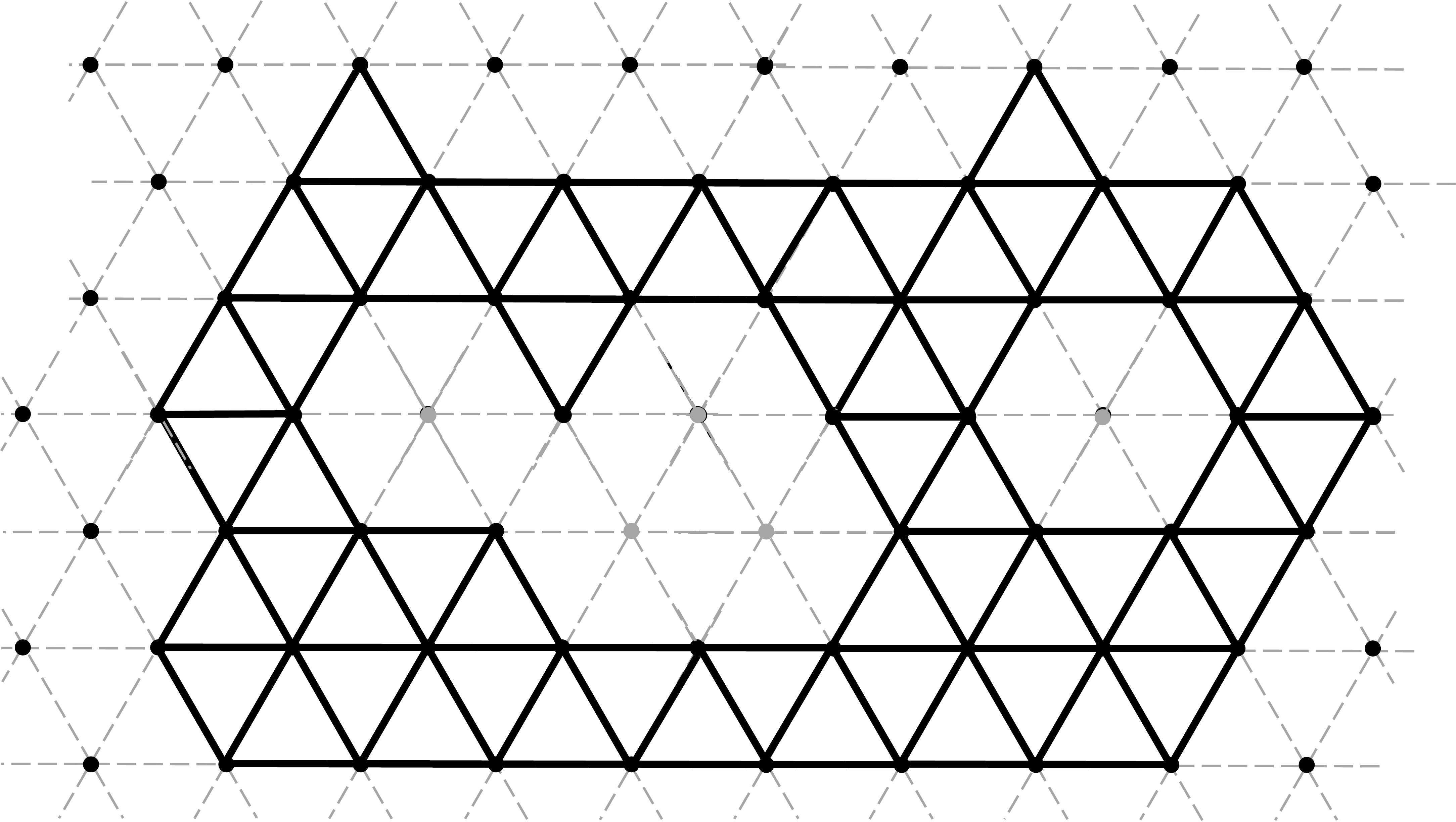}
  \caption{A triangular grid graph with $2$ holes.}
  \label{fig:grid_graph}
  \end{minipage}
\end{figure}

\section{Preliminaries}\label{sec:pre}

  Let $G=(V, E)$ be an undirected graph with $2n+1$ vertices.
  For a vertex $u$, we denote by $N(u)$ the set of vertices adjacent to $u$.
  For a vertex subset $X$, the subgraph induced by $X$ in $G$ is denoted by $G[X]$.
  A path or a cycle is \textit{odd} if it has an odd number of edges.

A \textit{matching} is a subset of edges that have no common end vertices.
A matching is \textit{nearly perfect} if its size is $n$.
A vertex is \textit{covered by} a matching $M$ if it is the end vertex of some edge in $M$, and \textit{exposed by} $M$ if it is not covered by $M$.
A cycle is \textit{$M$-alternating} if edges in $M$ and $E\setminus M$ appear alternatively along $C$, except for one vertex~(when the cycle is odd).

\paragraph*{Reconfiguration of Labeled Matching in Triangular Grid Graphs}

Consider the $2$-dimensional triangular lattice of infinite size.
A \textit{triangular grid} graph is a subgraph induced by a finite number of points in the triangular lattice.
See Figure~\ref{fig:grid_graph} for example.
In this paper, we also assume that a triangular grid graph is always connected.
A \textit{hole} of a triangular grid graph is a face of the graph whose boundary is a cycle of length at least $6$.

We here formally define our reconfiguration problem.
Let $G=(V, E)$ be a triangular grid graph with $2n+1$ vertices.
We denote $V=[2n+1]$, where, for a positive integer $x$, we write $[x]=\{1,2,\dots, x\}$.

A \textit{placement} is a mapping $p:[n]\to E$ such that $p(i)$ and $p(j)$ have no common end vertices for every distinct $i, j$.
We call each $p(i)$ a \textit{piece}.
Then $\{p(i)\mid i\in [n]\}\subseteq E$ forms a nearly perfect matching in $G$, which is denoted by $M_p$.
Let $v_p$ be the unique vertex exposed by $M_p$.
We also say that a mapping $p$ \textit{exposes} $v_p$.

We define the following operations on a placement, which we call \textsf{slide}~(see Figure~\ref{fig:slide}).
Suppose that there exists an integer $j\in [n]$ such that $p(j)=(u, v)$ and $(v, v_p)\in E$.
Then we transform $p$ to a placement $p_{\mathrm{s}}$ defined as 
\[
p_{\mathrm{s}}(i) = 
\begin{cases}
p(i) & (i\neq j)\\
(v, v_p) & (i=j).
\end{cases}
\]
The obtained placement $p_{\mathrm{s}}$ exposes the vertex $u$.
In this case, we write 
$p \leadsto p_{\mathrm{s}}$.

  \begin{figure}
    \centering
    \includegraphics[keepaspectratio,width=0.9\textwidth]
    {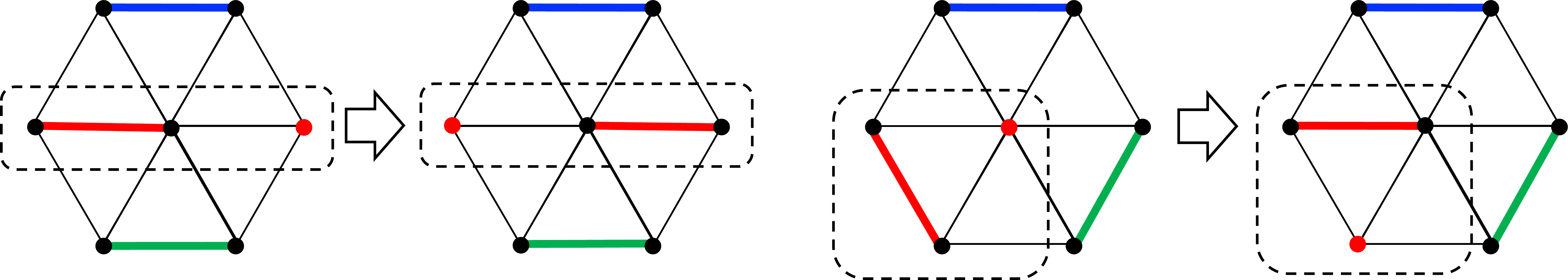}
    \caption{\textsf{Slide} operations. The colored, thick edges correspond to pieces.}
    \label{fig:slide}
  \end{figure}

Let $p, q$ be two distinct placements.
If there exists a sequence of placements
$\mu_0, \mu_1, \dots, \mu_{\ell}$ such that
(1) $\mu_0 = p$, $\mu_{\ell} = q$,
(2) $\mu_t \leadsto \mu_{t+1}$ for every integer
$t \in \{0,1,\dots,{\ell}-1\}$,
then we say that $p$ is \textit{reconfigured} to $q$.
A graph is \textit{reconfigurable} if any two placements can be reconfigured to each other.

We remark that, in the Gourds puzzle, a piece has a pair of labels~(numbers), meaning that each piece has an orientation.
That is, a mapping is defined from $[n]$ to $\{(u, v), (v, u)\mid (u, v)\in E\}$.
This requires us to define another operation to change the orientation of pieces.
Specifically, when a piece with the exposed vertex induces a triangle, we are allowed to change the orientation of the piece.
Our problem does not distinguish $(u, v)$ and $(v, u)$, and a placement is defined on a mapping from $[n]$ to $E$.
It should be noted, however, that our results can be adapted to the Gourds puzzle case with orientation.
See Sections~\ref{sec:fc_overview} and~\ref{sec:locally_connect} for the details.

\paragraph*{Rotation along a Cycle}

We define a sequence of \textsf{slide} operations, called \textit{rotation}, which will be used in the subsequent sections.
Let $C$ be an odd cycle.
We say that a placement $p$ is \textit{aligned with $C$} if $C$ is an odd $M_p$-alternating cycle and $C$ has the exposed vertex $v_p$.

  \begin{figure}
    \centering
    \includegraphics[keepaspectratio,width=0.9\textwidth]
    {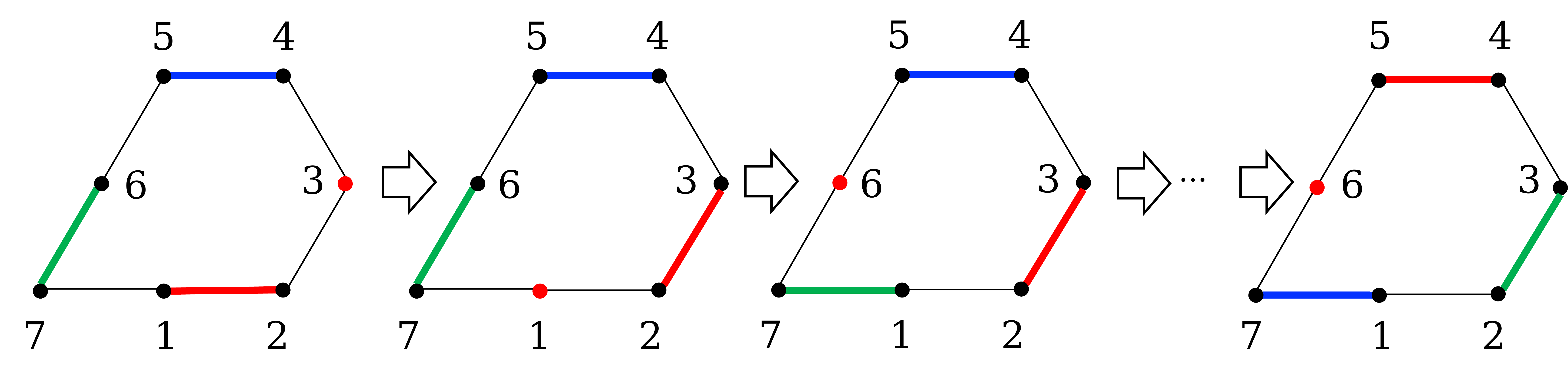}
    \caption{Rotation operations for a placement aligned with an odd cycle when $k=3$.}
    \label{fig:rotation}
  \end{figure}

Let $p$ be a placement aligned with $C$.
In what follows, 
we assume for simplicity that $V(C)=[2k+1]$ for some integer $k\in [n]$, and that the vertices of $C$ are aligned in the anti-clockwise order along $C$.
We also assume that the first $k$ pieces $p(1), p(2), \dots, p(k)$ of $p$ are located on the cycle $C$.

For an odd integer $j\in [2k+1]$, we define a placement $p_j$ as, for $i\in [k]$, 
\[
p_j(i) =
\begin{cases}
(2i-1, 2i) & (2i<j)\\
(2i, 2i+1) & (j<2i),
\end{cases}
\]
and $p_j(i)=p(i)$ for $i\geq k+1$.
Thus $p_j$ exposes the vertex $j$.
Moreover, for two integers $h, j\in[2k+1]$ such that $h\equiv j$~(mod $2$), define $p_{j, h}$ as, for $i\in [k]$, 
\[
p_{j, h}(i) =
\begin{cases}
(h+2i-2, h+2i-1) & (h+2i-1<j)\\
(h+2i-1, h+2i) & (j<h+2i-1),
\end{cases}
\]
where these vertex labels are defined on $\mathbb{Z}_{2k+1}$~(i.e., modulo $2k+1$), 
and $p_{j,h}(i)=p(i)$ for $i\geq k+1$.

Figure~\ref{fig:rotation} is an example when $k=3$.
The left-most figure depicts a placement $p_{3, 1}=p_3$ where 
$\left(p_{3, 1}(1), p_{3, 1}(2), p_{3, 1}(3)\right)=\left((1,2), (4,5), (6,7)\right)$.
By applying \textsf{slide} to $p_{3,1}$ once, we obtain $p_{1,1}=p_1$, that is,
$\left(p_{1,1}(1), p_{1,1}(2), p_{1,1}(3)\right)=\left((2,3), (4,5), (6,7)\right)$.
The right-most figure depicts a placement $p_{6,4}$, which is 
$\left(p_{6,4}(1), p_{6,4}(2), p_{6,4}(3)\right)=\left((4,5), (7,1), (2,3)\right)$.

The following observation asserts that $p_{j, h}$'s can be reconfigured to each other in $O(k^2)$ \textsf{slide} operations along $C$.
Such a sequence of \textsf{slide} operations is called \textit{rotation along $C$}, or we say that we \textit{rotate $p$ along $C$}.

\begin{observation}\label{obs:rotation}
For any two odd integers $j, j'\in [2k+1]$, 
we can reconfigure $p_j$ to $p_{j'}$ using at most $k$ \textsf{slide} operations.
Moreover, for any four integers $j, j'\in [2k+1]$ and $h, h'\in [2k+1]$ such that $j\equiv h$ and $j'\equiv h'$~\textup{(}$\textup{mod}\ 2$\textup{)},
we can reconfigure $p_{j, h}$ to $p_{j', h'}$ using at most $k^2+k$ \textsf{slide} operations.
\end{observation}

\begin{proof}
Observe that, applying \textsf{slide} to $p_j$ along $C$,
the exposed vertex $j$ is moved to $j-2$ or $j+2$~(mod $2k+1$).
This means that $p_j \leadsto p_{j-2}$ and $p_j \leadsto p_{j+2}$ hold.
Hence, by repeating \textsf{slide} operations at most $k$ times, $p_j$ can be reconfigured to $p_{j'}$.
We next show the second statement.
Similarly to the first statement, we can reconfigure $p_{j, h}$ to $p_{h, h}$ in at most $k$ \textsf{slide} operations.
Applying \textsf{slide} to $p_{h, h}$ along $C$, we obtain a placement $p_{h-2,h+1}$.
Repeating this procedure at most $k$ times, we can reconfigure $p_{j, h}$ to $p_{h'-3, h'}$.
Since we can reconfigure $p_{h'-3, h'}$ to $p_{j', h'}$ in at most $k$ \textsf{slide} operations, the total number of \textsf{slide} operations is at most $k^2+k$.
\end{proof}

\section{Reconfiguration on Factor-Critical Graphs}
\label{sec:factor-critical}

In this section, we discuss the reconfigurability of a factor-critical graph.
Recall that a graph is \textit{factor-critical} if, for any vertex $v$, $G$ has a nearly perfect matching that does not cover $v$.

As mentioned in Introduction, being a factor-critical graph is a necessary condition for reconfigurability.

\begin{observation}
If a triangular grid graph $G$ is reconfigurable, then it is factor-critical.
\end{observation}
\begin{proof}
If $G$ is not factor-critical, then $G$ has some vertex $u$ such that every nearly perfect matching covers $u$.
Then we cannot move the piece covering $u$ in an initial placement so that $u$ becomes exposed.
Hence there exist two placements such that one cannot be reconfigured to the other.
Thus the observation holds.
\end{proof}

Moreover, as observed in Hamersma et al.~\cite{ref:gourds}, the $2$-connectivity is necessary for a graph to be reconfigurable.
We remark that, even if a graph is  $2$-connected and factor-critical, it may not be reconfigurable. 
See Figure~\ref{fig:cant_example}.

\begin{figure}[tbp]
  \begin{minipage}[b]{0.48\columnwidth}
    \centering
\includegraphics[keepaspectratio,width=0.6\textwidth]
  {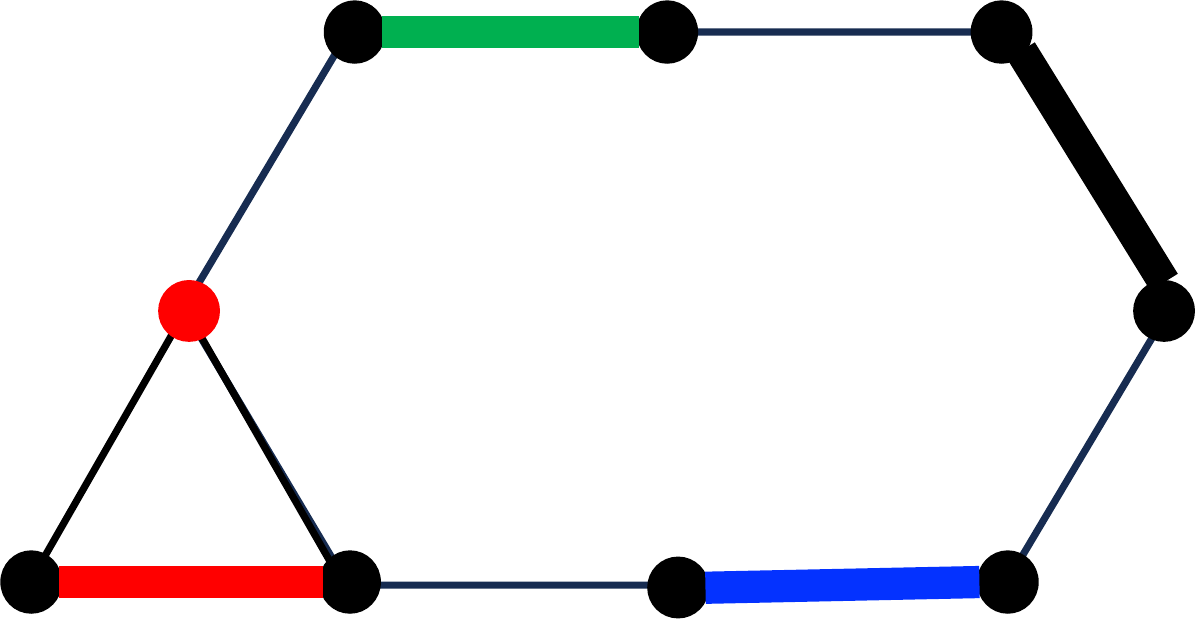}
  \caption{A factor-critical graph that is not reconfigurable. We cannot change the ordering of the pieces by \textsf{slide}.}
  \label{fig:cant_example}
  \end{minipage}
  \hspace{0.02\columnwidth} 
  \begin{minipage}[b]{0.48\columnwidth}
    \centering
\includegraphics[keepaspectratio,width=0.6\textwidth]
  {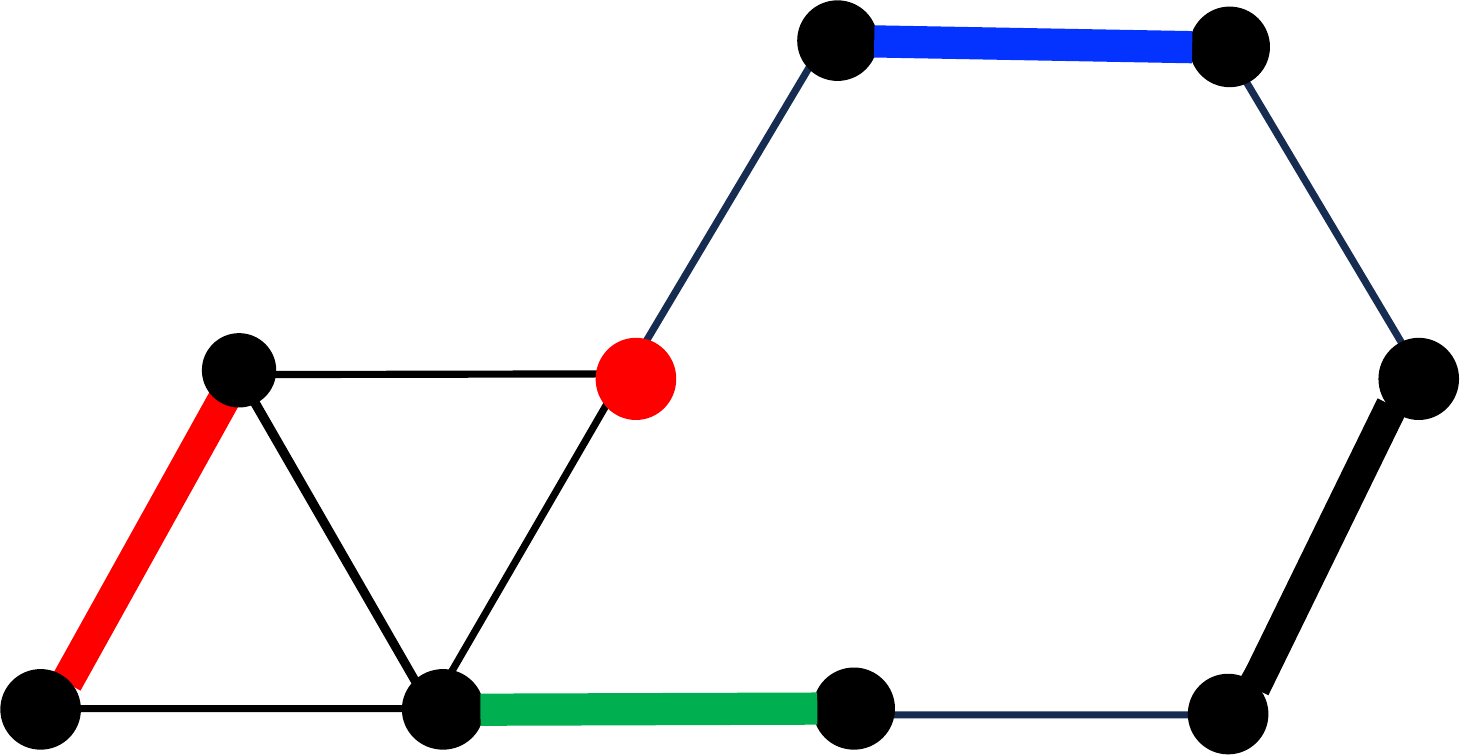}
  \caption{A factor-critical graph that is reconfigurable.}
  \label{fig:can_example}
  \end{minipage}
\end{figure}

The main theorem of this section is the following.
We show that a graph is reconfigurable if it has a vertex of degree $6$, which corresponds to a vertex not on the boundary cycles of the holes or the outer face.

\begin{restatable}{theorem}{factoralgorithm}
  \label{thm:factoralgorithm}
Let $G=(V, E)$ be a $2$-connected  factor-critical triangular grid graph.
If $G$ has a vertex of degree $6$, then $G$ is reconfigurable.
\end{restatable}

We remark that our condition is not necessary, as there exists a $2$-connected factor-critical graph such that it does not have a vertex of degree $6$, but it is reconfigurable. 
See Figure~\ref{fig:can_example}~(see also Lemma~\ref{lem:ThreeEarBase} and Section~\ref{sec:conclusion}).

\subsection{Proof Overview}\label{sec:fc_overview}

In this section, we present the proof overview of Theorem~\ref{thm:factoralgorithm}.
The proof makes use of the ear decomposition of a factor-critical graph to design a reconfiguration sequence.

An \textit{ear decomposition} of a graph $G$ is a sequence of subgraphs $G_1,G_2,\dots,G_k=G$ starting from a subgraph $G_1$ such that $G_{i+1}$ is obtained from $G_{i}$ by adding an ear $P_i$ for each $i\geq 1$, where an \textit{ear} $P$ of a subgraph $G'$ is a path of $G$ with end vertices in $G'$ such that $P$ is internally disjoint from $G'$.
We denote by $G'+P$ the subgraph obtained from $G'$ by adding the ear $P$.
Thus, in the ear decomposition, it holds that $G_{i+1}=G_i+P_i$ for each $i\in[k-1]$.
See Figure~\ref{fig:ear_align} for an example.

An ear decomposition is \textit{proper} if the end vertices of each ear are distinct, and \textit{odd} if each ear is of odd length.
It is known that a $2$-connected factor-critical graph is characterized by odd and proper ear decomposition.

\begin{proposition}[Theorem 5.5.2 in Lov\'{a}sz--Plummer \cite{ref:ear_decomposition}]
  \label{prop:2connect_critical}
  A graph $G$ is $2$-connected and factor-critical if and only if $G$ has an odd and proper ear decomposition starting from an odd cycle.
\end{proposition}

Let $p$ be a placement of $G$.
Recall that $M_p$ denotes a nearly perfect matching $\{p(i)\in E\mid i\in [n]\}$, and $v_p$ is the vertex exposed by $M_p$.
We say that a placement $p$ is \textit{aligned with an ear decomposition $G_1,\ldots,G_k$} if it satisfies the following two conditions~(Figure~\ref{fig:ear_align}).
  \begin{itemize}
    \item[(a)] $G_1$ is an odd $M_p$-alternating cycle with the exposed vertex $v_p$.
    \item[(b)] For each $i\in [k-1]$, the ear $P_i$ is $M_p$-alternating and its end vertices are not covered by $M_p\cap E(P_i)$.
  \end{itemize}

  \begin{figure}
    \centering
    \includegraphics[keepaspectratio,width=0.35\textwidth]
    {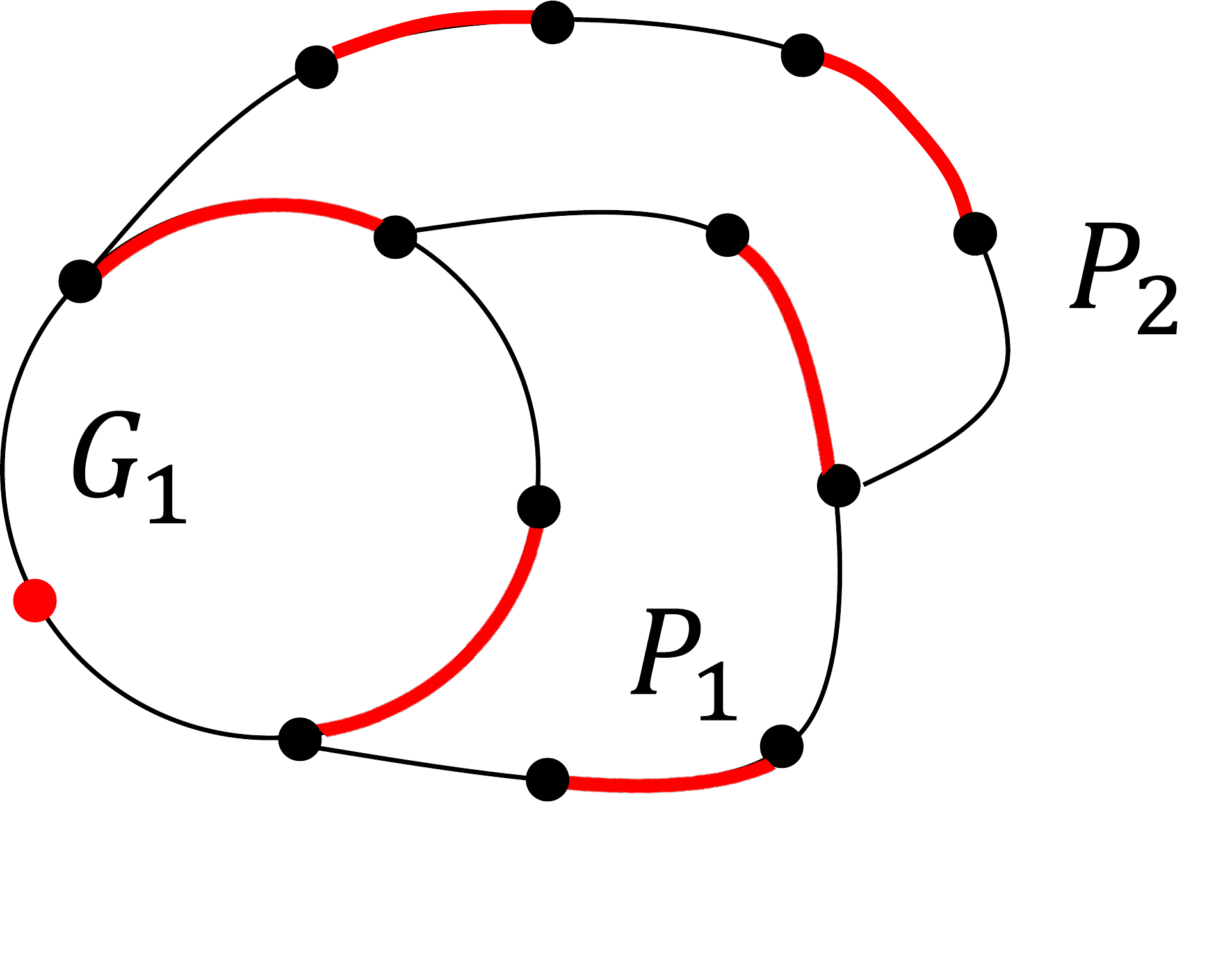}
    \caption{An ear decomposition and a matching aligned with the ear decomposition.}
    \label{fig:ear_align}
  \end{figure}

We show that any placement $p$ can be reconfigured so that the obtained placement $p'$ is aligned with a given ear decomposition $G_1, \dots, G_k$.
See Section~\ref{sec:Step1} for the proof.

\begin{restatable}{lemma}{alignedear}
\label{lem:alignedear}
  Let $G$ be a $2$-connected factor-critical triangular grid graph with $2n+1$ vertices.
  Let $G_1,\ldots,G_k$ be an odd and proper ear decomposition of $G$.
  Then we can reconfigure any placement $p$ to a placement aligned with $G_1,\ldots,G_k$ in $O(kn)$ \textsf{slide} operations.
\end{restatable}

Let $q$ be a target placement of $G$.
Applying Lemma~\ref{lem:alignedear} to $q$ as well, we can reconfigure $q$ so that the obtained placement $q'$ is aligned with $G_1, \dots, G_k$.
By taking the inverse of the reconfiguration steps, we see that $q'$ can be reconfigured to $q$ in $O(kn)$ \textsf{slide} operations.
Therefore, in order to reconfigure $p$ to $q$, it suffices to reconfigure $p'$ to $q'$.

We now present how to find a reconfiguration sequence between two placements aligned with $G_1, \dots, G_k$.
Since $G_i$ is factor-critical for any $i\in [k]$, 
the ear structure suggests to design a reconfiguration sequence recursively.
In fact, we will show in Lemma~\ref{lem:phase2}~(Section~\ref{subsec:year_algorithm}) that, if $G_j$ is a reconfigurable graph with at least $5$ vertices for some $j < k$, then so is $G_k=G$.
Note that the lemma holds even if a graph is not a triangular grid graph.
However, $G_j$'s may not necessarily be reconfigurable, as there exists a factor-critical graph which is not reconfigurable.
To overcome the difficulty, we introduce a special kind of ear decomposition starting from simple reconfigurable subgraphs.

We say that an odd and proper ear decomposition $G_1, G_2, \dots, G_k$ is \textit{admissible} if it satisfies either 
\begin{itemize}
\item[(i)] $G_1$ is a cycle of length $5$~(Figure~\ref{fig:pentagon_ear}), or
\item[(ii)] $P_1$ is of length $3$ and has the end vertices $u, v$ which are adjacent in $G_1$~(Figure~\ref{fig:3_ear}).
\end{itemize}

Consider the case when (i) is satisfied.
Since the ordering of ears with length $1$ may be changed in the ear decomposition, we may assume that the first $2$ ears $P_1$ and $P_2$ are the inner edges of $G_1$.
Then $G_3 = G_1+P_1+P_2$ induces a pentagon in $G$, where a \textit{pentagon} is a subgraph induced by three adjacent triangles.
It is not difficult to see that the pentagon is reconfigurable in a constant number of \textsf{slide} operations.
On the other hand, consider the case when (ii) is satisfied.
Similarly to the case~(i), we may assume that the second ear $P_2$ is the inner edge of $P_1$.
The subgraph $G_3 = G_1+P_1+P_2$ induces an odd cycle attached to a diamond, where a \textit{diamond} is a subgraph induced by two adjacent triangles.
The subgraph is shown to be reconfigurable in Lemma~\ref{lem:ThreeEarBase}. 

\begin{figure}[t]
  \begin{minipage}[b]{0.48\columnwidth}
    \centering
\includegraphics[keepaspectratio,width=0.45\textwidth]
    {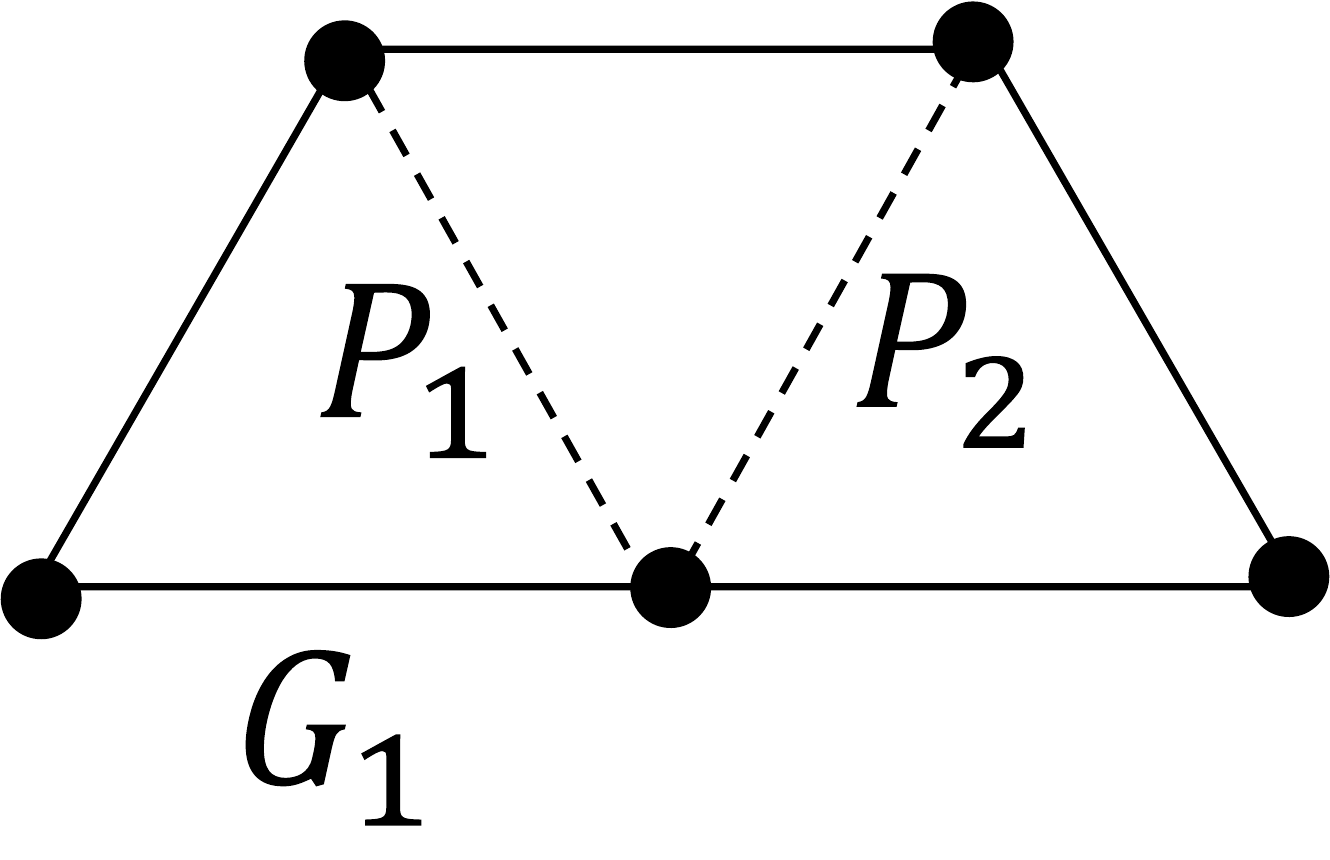}
    \caption{A pentagon.}
    \label{fig:pentagon_ear}
  \end{minipage}
  \hspace{0.02\columnwidth} 
  \begin{minipage}[b]{0.48\columnwidth}
    \centering
    \includegraphics[keepaspectratio,width=0.8\textwidth]
    {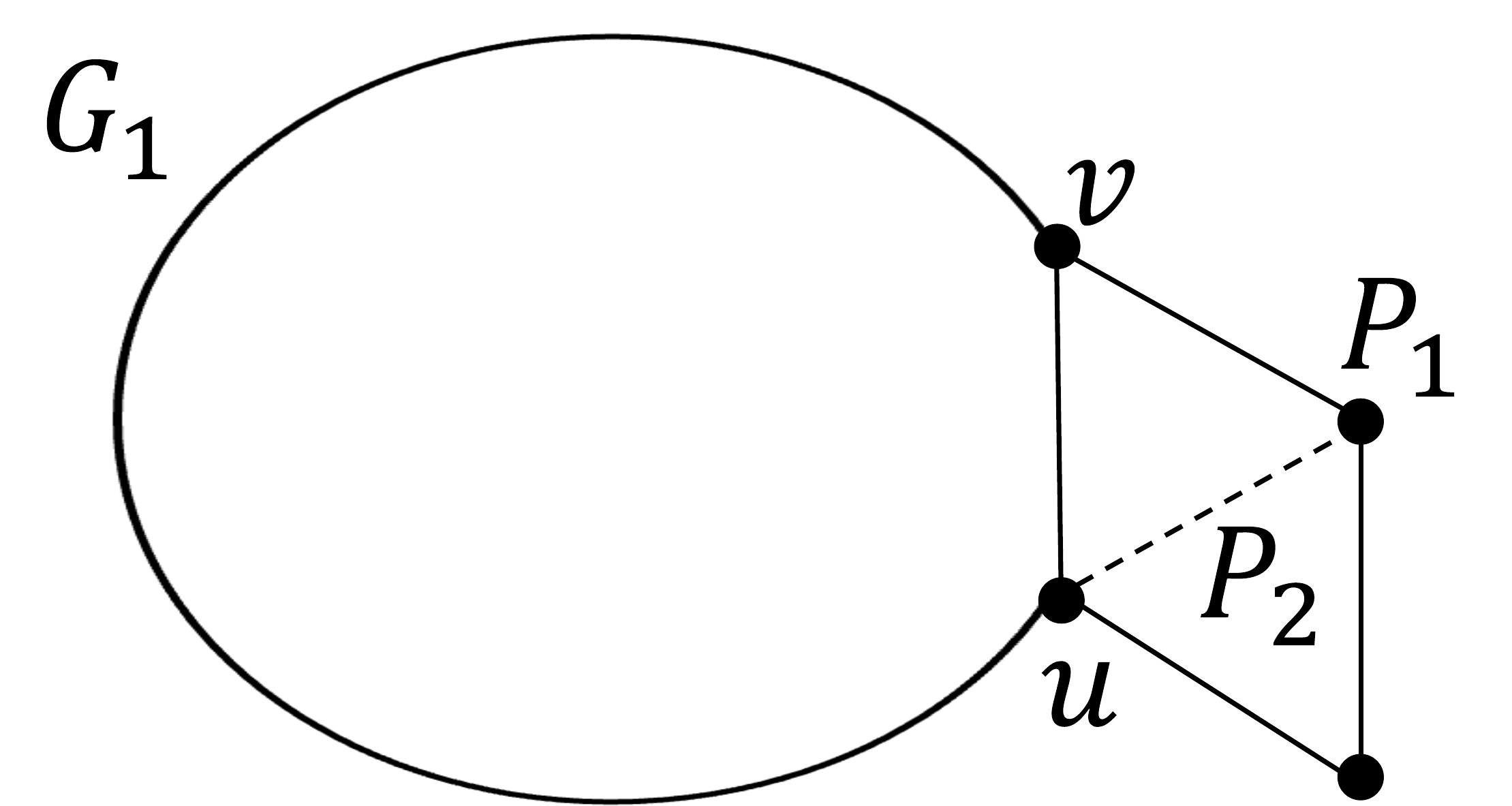}
    \caption{An odd cycle with a diamond.}
    \label{fig:3_ear}
  \end{minipage}
\end{figure}

Therefore, $G_3$ in an admissible ear decomposition is reconfigurable.
Hence, if there exists an admissible ear decomposition, we can find a reconfiguration sequence as below.
Section~\ref{subsec:year_algorithm} for the details.

\begin{restatable}{lemma}{reconfear}
\label{lem:reconfear}
Let $G_1, \dots, G_k=G$ be an admissible ear decomposition of $G$.
Then any two placements aligned with the ear decomposition can be reconfigured to each other.
\end{restatable}

We remark that the length of a reconfiguration sequence obtained in the above lemma is at most $n^{2n}$ for a graph with $2n+1$ vertices, which is not bounded by a polynomial in $n$.
It may be interesting to find the optimal bound on the length of reconfiguration sequences.
It is known in~\cite{ref:gourds} that the length is $\Omega (n^2)$.

Finally, we show that there always exists an admissible ear decomposition in a triangular grid graph with a vertex of degree $6$.

\begin{theorem}
  \label{thm:admissible}
  Let $G$ be a $2$-connected factor-critical triangular grid graph such that it has a vertex of degree $6$.
Then $G$ has an admissible ear decomposition.
  \end{theorem}

Theorem~\ref{thm:admissible} is a graph-theoretical result independent of designing a reconfiguration sequence.
Theorem~\ref{thm:admissible} can be proved by investigating the matching structure of factor-critical triangular grid graphs. 
The proof will be given in Section~\ref{sec:admissible_eardecomposition}.

In summary, a reconfiguration sequence from an initial placement $p$ to a target placement $q$ can be realized as below, which completes the proof of Theorem~\ref{thm:factoralgorithm}.

\begin{enumerate}
  \item Reconfigure $p$ to a placement aligned with an admissible ear decomposition $G_1,\ldots,G_k$, denoted by $p'$, by Lemma~\ref{lem:alignedear}.
  \item Using Lemma~\ref{lem:reconfear}, reconfigure $p'$ to another placement $q'$ aligned with  $G_1,\ldots,G_k$, where $q'$ is a placement obtained from $q$ by Lemma~\ref{lem:alignedear}.
  \item Reconfigure $q'$ to the target placement $q$.
\end{enumerate}

We remark that the proof of Theorem~\ref{thm:factoralgorithm} above can be adapted to the Gourds puzzle in which a piece has an orientation.
This is because the structures (i) and (ii) in an admissible ear decomposition can also be used to change the orientation of pieces in an arbitrary way.
Thus we have the following corollary.

\begin{corollary}
  \label{cor:gourds}
Let $B$ be a hexagonal grid such that the dual triangular grid graph is a $2$-connected factor-critical graph with a vertex of degree $6$.
Then any two configurations of the same set of pieces on $B$ can be reconfigured to each other.
\end{corollary}

\subsection{Reconfiguration to a Placement Aligned with Ear Decomposition}\label{sec:Step1}

In this subsection, we will show Lemma~\ref{lem:alignedear}, that is, we will show that we can reconfigure an initial placement $p$ to a placement aligned with a given odd and proper ear decomposition $G_1, \dots, G_k$.

We first prove that we can reconfigure so that any vertex is exposed. 

\begin{lemma}
  \label{lem:exposed}
  Let $G$ be a $2$-connected factor-critical triangular grid graph with $2n+1$ vertices.
    For any vertex $v$, we can reconfigure a placement $p$ so that $v$ is the exposed vertex, in $O(n)$ \textsf{slide} operations.
\end{lemma}

\begin{proof}
  Since $G$ is factor-critical, $G$ has a nearly perfect matching $M_v$ that exposes $v$.
  The symmetric difference $M_p \Delta M_v$ contains an $M_p$-alternating path $P$ from $v_p$ to $v$, which is of even length. 
  We reconfigure $p$ by sliding the pieces on the path $P$ one-by-one.  
  The resulting placement exposes $v$.
  The number of \textsf{slide} operations is $|P|/2$, which is $O(n)$.
\end{proof}

To obtain a placement aligned with the ear decomposition, 
we first reconfigure so that an end vertex of the last ear $P_{k-1}$ is exposed using Lemma~\ref{lem:exposed}.
Then, since each inner vertex of the last ear $P_{k-1}$ is of degree $2$ in $G$, the obtained placement is aligned with $P_{k-1}$.
By applying this procedure repeatedly for each ear, we can obtain a placement aligned with $G_1, \dots, G_k$.
This implies  Lemma~\ref{lem:alignedear} as below.

\begin{proof}[Proof of Lemma~\ref{lem:alignedear}]
  Let $v_i$ be an end vertex of ear $P_i$ for $i\in [k-1]$.
The basic observation is that, each inner vertex of the last ear $P_{k-1}$ is of degree $2$, and hence, if a nearly perfect matching $M$ exposes $v_{k-1}$, the last ear $P_{k-1}$ is an $M$-alternating path such that the end vertices are not covered by edges of $M\cap E(P_{k-1})$.

  We perform the following procedure for each $i=k-1,k-2,\ldots,1$.
  Note that $G_i$ is factor-critical for any $i\in [k]$.
  \begin{enumerate}
    \item Applying Lemma~\ref{lem:exposed} to $G_{i+1}$, we reconfigure the current placement of $G_{i+1}$ so that $v_i$ is exposed.
    Then the resulting placement is aligned with $P_i$ by the above observation.
  \end{enumerate}
  In the end of the above procedure, the obtained placement of the original graph $G$ is aligned with $P_{k-1},\ldots,P_1$.
  Moreover, the exposed vertex is on $G_1$.
  Thus this is a desired placement.
  The necessary number of \textsf{slide} operations is $O(kn)$, since we repeat the procedure of Lemma~\ref{lem:exposed} $k-1$ times.
  \end{proof}

\subsection{Reconfiguration using Ear Decomposition}
\label{subsec:year_algorithm}

We next present how to reconfigure a placement aligned with an ear decomposition.
Using the ear structure,  we can find a reconfiguration sequence if the subgraph $G_{k-1}$ is reconfigurable.

\begin{restatable}{lemma}{PhaseTwo}
  \label{lem:phase2}
Let $G_1, \dots, G_k$ be an odd and proper ear decomposition of a graph $G$ with $2n+1$ vertices.  
Suppose that $G_{k-1}$ has at least $5$ vertices, and that, in $G_{k-1}$, any placement aligned with the ear decomposition $G_1, \dots, G_{k-1}$ can be reconfigured to another placement aligned with $G_1, \dots, G_{k-1}$, using $t$ \textsf{slide} operations.
Then there exists a reconfiguration sequence between any two placements along with $G_1, \dots, G_k$ in $G$, which requires $O(n^2 (t+n))$ \textsf{slide} operations.
\end{restatable}

            \begin{figure}
            \centering
            \includegraphics[keepaspectratio,width=0.9\textwidth]
            {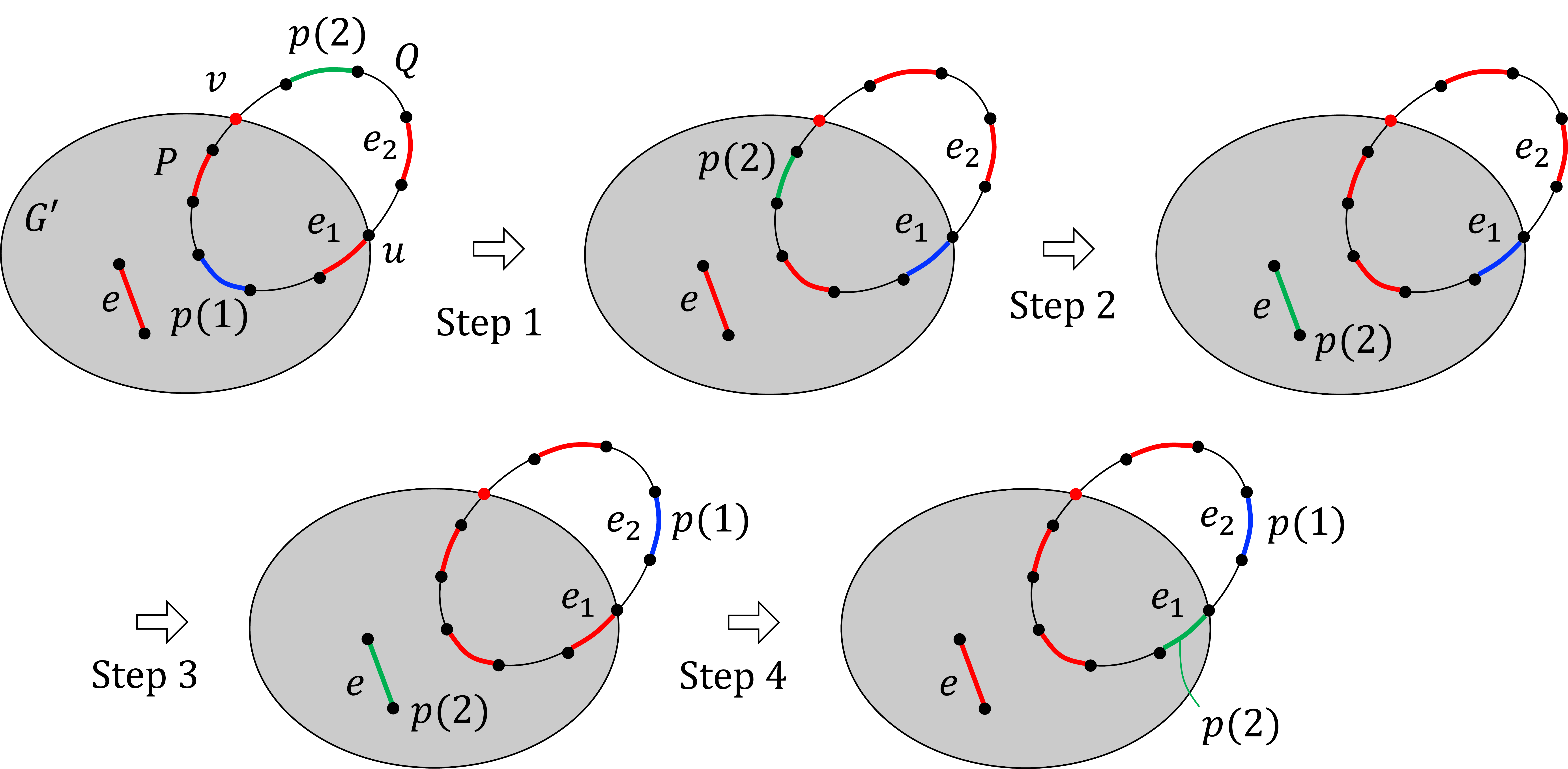}
            \caption{Proof of Claim~\ref{clm:nonHamilton}: An example when $j=2$.}
            \label{fig:YP_not_hamilton}
          \end{figure}

\begin{proof}
  For simplicity, we denote $G_{k-1}=G'$ and $P_{k-1}=Q$ in the proof.
  Let $u, v$ be the end vertices of $Q$.
  Let $p$ be an initial placement of $G$ and $q$ be a target placement of $G$, both of which are aligned with the ear decomposition.
  It follows from Lemma~\ref{lem:exposed} that we can reconfigure $p$~(and $q$, resp.) so that $v$ is exposed.
  Hence we may assume that both $p$ and $q$ expose $v$.
    Moreover, by changing the indices of the pieces if necessary, we may assume that the pieces $q(1), \dots, g(\ell)$ are placed on the ear $Q$ in the order from $v$ to $u$.

\begin{claim}
There exists an $M_p$-alternating path $P$ of even length from $v$ to $u$ in $G'$.
\end{claim}
\begin{proof}
  Since $G'$ is factor-critical, there exists a nearly perfect matching $M_u$ that exposes $u$.
  Taking $M_p\triangle M_u$, we see that there exists 
an even $M_p$-alternating path $P$ from $v$ to $u$ in $G'$.
\end{proof}

Let $C$ be the cycle consisting of $P$ and $Q$.
Then $C$ is an odd $M$-alternating cycle in $G$.
  We will show that we can reconfigure $p$ so that the first $\ell$ pieces are on the ear $Q$, using the cycle $C$.
  We consider the following two cases, depending on whether $C$ is a Hamilton cycle or not.

  \begin{claim}\label{clm:nonHamilton}
  Suppose that $C$ is not a Hamilton cycle of $G$.
  Then we can reconfigure $p$ to a placement $p'$ so that $p'(1), \dots, p'(\ell)$ are placed on $Q$ in this order from $v$ to $u$, using $O(\ell (t+n^2))$ \textsf{slide} operations.
  \end{claim}
  \begin{proof}
      In this case, the graph $G'$ has an edge $e\not\in E(P)$ such that $e\in M_p$.
      Let $e_1$ be the edge of $M_p$ covering $u$, and $e_2$ be the edge of of $M_p$ covering the vertex adjacent to $u$ on $Q$.
      See Figure~\ref{fig:YP_not_hamilton}.

      We may assume that $p(1)$ is on the cycle $C$, as otherwise $p(1)$ is contained in $G'$, and hence we can reconfigure the current placement in $G'$ so that $p(1)$ is on $P$, keeping $v$ exposed, using $t$ \textsf{slide} operations by the assumption.

      We reconfigure $p$ by the following $4$ steps for $j=2,\ldots, \ell$.
      Initially, we set $\tilde{p}=p$.

          \begin{enumerate}
            \item We move the piece $\tilde{p}(j)$ so that $\tilde{p}(j)$ is  on $P$ as follows.
            \begin{enumerate}
            \item If $\tilde{p}(j)$ is on $Q$, we rotate the current placement $\tilde{p}$ along $C$ so that $\tilde{p}(j)$ is located on $P$.
            \item If $\tilde{p}(j)$ is contained in $G'$ but not on $P$, then we reconfigure the current placement $\tilde{p}$ on $G'$ so that $\tilde{p}(j)$ is located on $P$, keeping that $\tilde{p}(1), \dots, \tilde{p}(j-1)$ are on $C$.
            \end{enumerate}
            
            \item We reconfigure the current placement $\tilde{p}$ on $G'$ to swap $\tilde{p}(j)$ and the piece on $e$. 
            Thus $\tilde{p}(j)=e$.
            \item We rotate the current placement $\tilde{p}$ along $C$ so that $\tilde{p}(j-1)$ is  $e_2$.
            \item We reconfigure the current placement $\tilde{p}$ on $G'$ to swap $\tilde{p}(j)$ and the piece on $e_1$.
            Thus $\tilde{p}(j)$ has been changed to $e_1$. 

            In the end of the $j$-th iteration, $\tilde{p}(1), \dots, \tilde{p}(j)$ are located on $C$ in this order from $v$ to $u$.
          \end{enumerate}
          Therefore, in the end of the above procedure, the pieces $\tilde{p}(1), \dots, \tilde{p}(\ell)$ are located on $C$ in this order from $v$ to $u$.
          Thus we can rotate $\tilde{p}$ along $C$ so that they are on $Q$. 
          
          In the above procedure, for each $j$, we reconfigure the placement restricted on $G'$ in a constant number of times, and we rotate the placement along $C$ at most twice.
          Therefore, the total number of \textsf{slide} operations is $O(\ell (t+n^2))$ by Observation~\ref{obs:rotation}.
  \end{proof}

  \begin{claim}
  Suppose that $C$ is a Hamilton cycle of $G$.
  Then we can reconfigure $p$ to a placement $p'$ so that $p'(1), \dots, p'(\ell)$ are placed on $Q$ in this order from $v$ to $u$, using $O(\ell n(t+n))$ \textsf{slide} operations.
  \end{claim}
  \begin{proof}
          Since $G'$ has at least $5$ vertices, $P$ has at least $2$ edges of $M_p$.
          Let $e_1, e_2$ be two edges in $M_p\cap E(P)$ such that $e_1, e_2$ appear consecutively along $P$.
          We can swap the $2$ pieces on $e_1$ and $e_2$ by reconfiguring on $G'$, using $t$ \textsf{slide} operations.
          By using the strategy similar to the bubble sort algorithm, we can obtain a placement $p'$ such that  $p'(1),\ldots,p'(\ell)$ are on $C$ in this order from $v$ to $u$. This requires $O(\ell n)$ swaps. 
          Since each swap takes $O(t+n)$ \textsf{slide} operations, it takes 
          $O(\ell n(t+n))$ \textsf{slide} operations in total.  
  \end{proof}

  In each case, we can reconfigure $p$ to a placement $p'$ so that the pieces $p'(1),\ldots,p'(\ell)$ are located on $Q$ in this order from $v$ to $u$.
  Since we can reconfigure the placement on $G'$ to any placement, 
  we can reconfigure $p'$ to $q$.
  The total number of \textsf{slide} operations is $O(\ell n(t+n))=O(n^2(t+n))$.
\end{proof}

By applying Lemma~\ref{lem:phase2} recursively, 
we see that, if $G_j$ is a reconfigurable subgraph with at least $5$ vertices for some $j < k$, then $G_k=G$ is reconfigurable.
In particular, if a given ear decomposition is admissible, then $G$ is shown to be reconfigurable.

Below we upper-bound the number of operations to reconfigure two placements aligned with an admissible ear decomposition.
We will show each case of the definition of an admissible ear decomposition separately.
We assume that a graph has $2n+1$ vertices for $n\geq 2$.

\begin{restatable}{lemma}{pentagon}
  \label{lem:pentagon}
  Let $G_1, \dots, G_k$ be an admissible ear decomposition such that $G_1$ is a cycle of length $5$~(Figure~\ref{fig:pentagon_ear}).
  Then we can reconfigure an arbitrary placement $p$ aligned with $G_1,\ldots,G_k$ to another placement aligned with $G_1,\ldots,G_k$ in 
  at most $n^{2n}$ \textsf{slide} operations.
\end{restatable}
\begin{proof}
  For $i\in [k-1]$, we denote by $P_i$ the $i$-th ear in $G_1,\ldots,G_k$.
  The cycle $G_1$ of length $5$ induces a pentagon consisting of three triangles.
  We may assume that $P_1$ and $P_2$ are the inner edges of the pentagon~(Figure~\ref{fig:pentagon_ear}), as we may change the ordering of an ear of length one in the ear decomposition.
  The pentagon is a reconfigurable graph in a constant number of \textsf{slide} operations.
  Since the pentagon $G_3=G_1+P_1+P_2$ is reconfigurable, we can see that $G_i$ is reconfigurable for any $i\geq 3$ by applying Lemma~\ref{lem:phase2} recursively.
  
  We now estimate the number of \textsf{slide} operations for reconfiguration.
  Let $f(n)$ be the number of \textsf{slide} operations for a graph with $2n+1$ vertices where $n\geq 2$.
  Then it follows from Lemma~\ref{lem:phase2} that $f(n)\leq n^2 (f(n-\ell)+n)\leq n^2 f(n-1)+n^3$ for $n\geq 3$, where $2\ell+1$ is the length of $P_{k-1}$. 
  
  We will show that $f(n)\leq n^{2n}-2n$ for $n\geq 2$.
  Indeed, if $n=2$, then a graph is a pentagon and it is not difficult to see that the graph is reconfigurable in at most $8$ \textsf{slide} operations.
  When $n\geq 3$, it holds by induction that
  \begin{align*}
  f(n) &\leq n^2 \left((n-1)^{2(n-1)}-2(n-1)\right) + n^3 \leq n^{2n} - 2n^2(n-1)+n^3 \\
  &=n^{2n} - n (n^2 - 2n) \leq n^{2n} - 2n,
  \end{align*}
  where the last inequality holds since $n^2-2n-2\geq 0$ for $n\geq 3$.
  Thus $f(n)\leq n^{2n}-2n$ holds.
\end{proof}


We next discuss the second case of an admissible ear decomposition.
The following lemma says that the base case is reconfigurable.

Let $\tilde{G}=(V, E)$ be a triangular grid graph with $2n+1$ vertices for $n\geq 3$ as in Figure~\ref{fig:can_exampleLabel}.
More specifically, $V=[2n+1]$, and it consists of an odd cycle $C$ of length $2n-1$ with vertex set $[2n-1]$, attached to a diamond $D$ with vertex set $\{2n-2, 2n-1, 2n, 2n+1\}$.

\begin{figure}
  \centering
\includegraphics[keepaspectratio,width=0.5\textwidth]{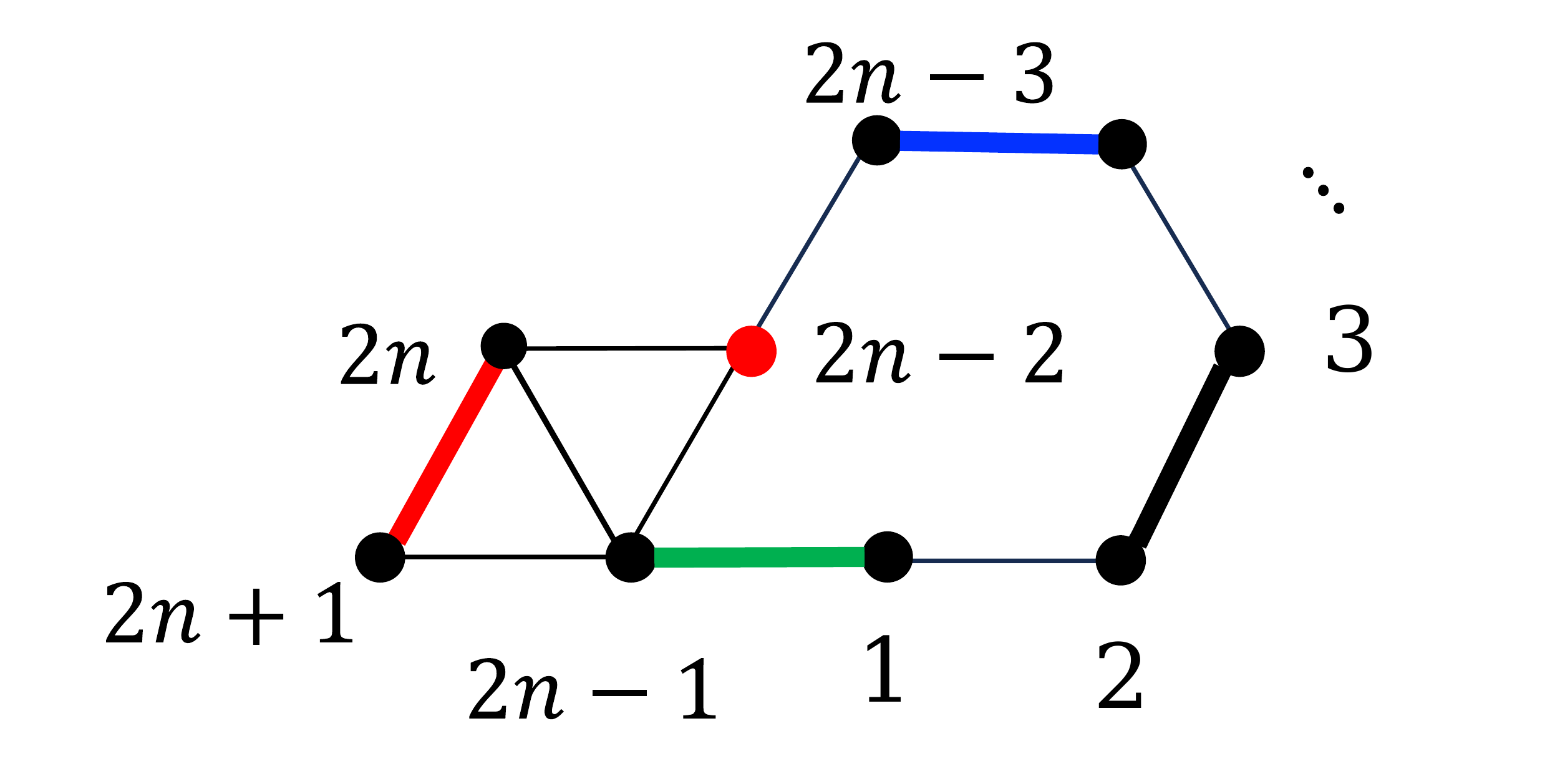}
  \caption{A factor-critical graph that is reconfigurable.}
  \label{fig:can_exampleLabel}
\end{figure}

\begin{lemma}\label{lem:ThreeEarBase}
The graph $\tilde{G}$ defined above with $2n+1$ vertices~($n\geq 3$) is reconfigurable in at most $n^3+n^2$ operations.
\end{lemma}
\begin{proof}
Let $p$ and $q$ be an initial and target placements of $\tilde{G}$, respectively.
We may assume that the target pieces $q(1), \dots , q(n-1)$ are located in the anti-clockwise order along $C$, and that $D$ has pieces $q(n-1)$ and $q(n)$.
Let $C'$ be the Hamilton cycle of length $2n+1$ in $\tilde{G}$.

We present a reconfiguration sequence as follows: 
Initially, we set $\tilde{p}=p$.
For $j=1, 2, \dots, n$, we do the following $2$ steps.

\begin{enumerate}
\item We rotate the current placement $\tilde{p}$ along $C'$ so that $\tilde{p}(j)$ is equal to the edge $(2n, 2n+1)$.
\item We rotate the current placement $\tilde{p}$ along $C$ so that $\tilde{p}(j-1)$ is equal to the edge $(2n-1, 1)$.
Then $\tilde{p}(1), \dots, \tilde{p}(j)$ are located on $C'$ in the anti-clockwise order.
\end{enumerate}

In the end of the procedure, $\tilde{p}(1), \dots, \tilde{p}(n)$ are located on $C'$ in the anti-clockwise order, which is the desired placement $q$.
In each iteration, we rotate the current placement along $C$ or $C'$.
By choosing the shorter one between the clockwise rotation and the anti-clockwise rotation, it requires at most $n^2+n$ \textsf{slide} operations by Observation~\ref{obs:rotation}.
Hence the total number of \textsf{slide} operations is at most $n^3+n^2$.
\end{proof}

We next show the case when an admissible ear decomposition satisfies the second case.
This, together with Lemma~\ref{lem:pentagon},  proves Lemma~\ref{lem:reconfear}.

\begin{restatable}{lemma}{threeear}
  \label{lem:3_ear}
  Let $G_1, \dots, G_k$ be an admissible ear decomposition such that $P_1$ is of length $3$ and has the end vertices $u, v$ which are adjacent in $G_1$~(Figures~\ref{fig:3_ear} and~\ref{fig:can_exampleLabel}).
  Then we can reconfigure an arbitrary placement $p$ aligned with $G_1,\ldots,G_k$ to another placement in 
  at most $n^{2n}$ \textsf{slide} operations.
\end{restatable}

\begin{proof}
  The edge $(v, u)$ with the path $P_1$ consists of a diamond with two triangles.
  We may assume that the diagonal line of the diamond is $P_2$.
Hence, by Lemma~\ref{lem:ThreeEarBase}, $G_3=G_2+P_2$ is reconfigurable in $O(t^2)$ operations, where $t$ is the number of vertices in $G_3$.
  By applying Lemma~\ref{lem:phase2} recursively, we can see that $G_i$ is reconfigurable for any $i\geq 3$, and hence so is $G=G_k$.
  In a way similar to the proof of Lemma~\ref{lem:pentagon}, 
  the total number of \textsf{slide} operations is bounded by $n^{2n}-2n$.
  Thus the lemma holds.
\end{proof}

\subsection{Proof of Theorem~\ref{thm:admissible}: Existence of Admissible Ear Decomposition}\label{sec:admissible_eardecomposition}

A subgraph $H$ of $G$ is called \textit{central} if $G-H$ has a perfect matching, where $G-H$ is the subgraph obtained by removing the vertices of $H$ from $G$.
Note that an odd cycle $C$ is central if and only if there exists a nearly perfect matching $M$ such that $C$ is an odd  $M$-alternating cycle that has an exposed vertex.

As remarked in~\cite{ref:ear_decomposition}, 
the proof of Proposition~\ref{prop:2connect_critical} shows the following statement, which says that any odd and proper ear decomposition of a subgraph can be extended to one for the whole graph.

\begin{lemma}[Theorem 5.5.2 in Lov\'{a}sz--Plummer \cite{ref:ear_decomposition}]\label{lem:centralcycle1}
  Let $G$ be a $2$-connected factor-critical graph, and $H$ be a central subgraph that has an odd and proper ear decomposition $G'_1,\ldots,G'_{k'}=H$. Then $G$ has an odd and proper ear decomposition $G_1,\ldots,G_k$ such that $G_i=G'_i$ for $i\in[k']$.
\end{lemma}

In addition, we observe the following lemma.

\begin{lemma}\label{lem:centralcycle2}
  Let $G$ be a $2$-connected factor-critical graph, and $M$ be a nearly perfect matching.
  Then, for any edge $e$ incident to the exposed vertex $v$, $G$ has an odd $M$-alternating cycle having $e$.
\end{lemma}
\begin{proof}
We denote $e=(u, v)$ for some vertex $u$ in $N(v)$.
Since $G$ is factor-critical, $G$ has a nearly perfect matching $M'$ exposing $u$.
Consider the subgraph with edge set $M\triangle M'$.
It contains an $M$-alternating path $P$ of even length from $u$ to $v$.
Hence, the path $P$ with the edge $(u, v)$ yields 
an $M$-alternating cycle of odd length in $G$.
Thus the lemma holds.
\end{proof}

The theorem below shows Theorem~\ref{thm:admissible}.

\begin{theorem}
  \label{thm:find_ear}
  Let $G$ be a $2$-connected triangular grid graph such that it has a vertex $o$ of degree $6$.
  If $G$ has a nearly perfect matching $M$ exposing $o$, then $G$ has an admissible ear decomposition.
  \end{theorem}

\begin{proof}
  We denote by $a,b,c,d,e,f$ the $6$ neighbor vertices of $o$ in the anti-clockwise way as in Figure~\ref{fig:3_hexagon}.
  These $7$ vertices induce a hexagon.
  We denote by $D$ the edge set of the boundary cycle of the hexagon.

  In what follows, we will prove that $G$ has either
  \begin{itemize}
\item[(i)] a central cycle of length $5$, or
\item[(ii)] an odd central cycle $C$ and a path $P$ of length $3$ such that $P$ has the end vertices adjacent in $C$ and the subgraph $C+P$ is central.
\end{itemize}
  Then we can extend by Lemma~\ref{lem:centralcycle1} these structures to an admissible ear decomposition of the whole graph $G$.
  
  In what follows, we consider several cases according to the number of matching edges $|D\cap M|$ on the boundary cycle of the hexagon.

\medskip
\noindent
\textbf{(a) $|D\cap M|=3$.}

By symmetry, edges of $D\cap M$ are aligned as in Figure~\ref{fig:3_hexagon}.
This clearly has an $M$-alternating cycle of length $5$, and thus the structure~(i) exists.
Thus $G$ has an admissible ear decomposition by Lemma~\ref{lem:centralcycle1}.

          \begin{figure}
            \centering
            \includegraphics[keepaspectratio,width=0.3\textwidth]
            {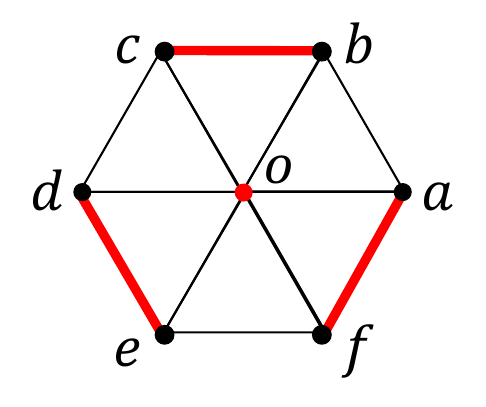}
            \caption{Case~(a) when $|D\cap M|=3$.}
            \label{fig:3_hexagon}
          \end{figure}

\medskip
\noindent
\textbf{(b) $|D\cap M|=2$.}

          By symmetry, we may assume that 
          either $D\cap M=\{(b, c), (d, e)\}$ or $D\cap M=\{(b,c), (e,f)\}$ as depicted in Figure~\ref{fig:2_hexagon}.
          When $D\cap M=\{(b, c), (d, e)\}$, it has a central cycle of length $5$, and hence the structure~(i) exists.
          We then consider the case when $D\cap M=\{(b,c), (e,f)\}$.
          Since the vertices $a$ and $d$ are covered by $M$, they are adjacent to some vertices $h$ and $g$, respectively, with edges of $M$.
          Note that we may assume that the edges $(d, h)$ and $(a, g)$ are horizontal edges as in Figure~\ref{fig:2_hexagon_G1}, as otherwise there would exist a central cycle of length $5$.
          
          \begin{figure}
            \centering
            \includegraphics[keepaspectratio,width=0.6\textwidth]
            {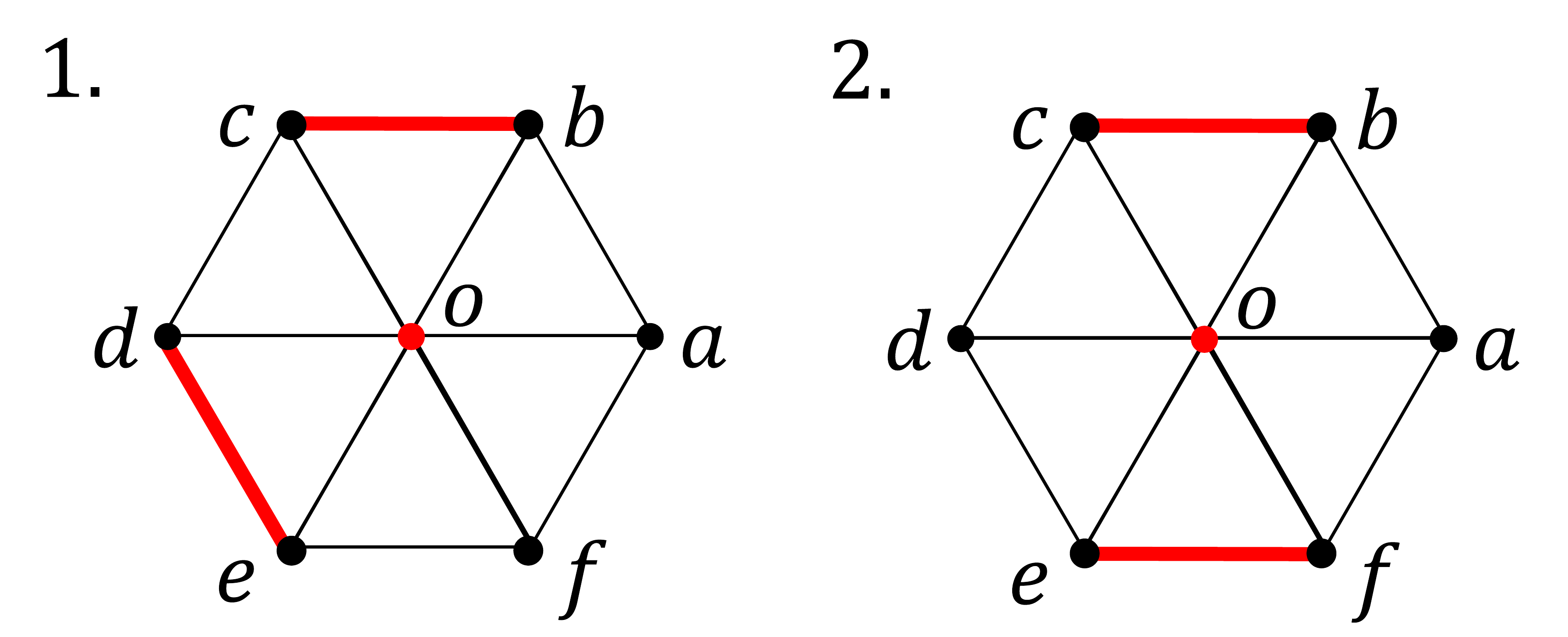}
            \caption{Case~(b) when $|D\cap M|=2$.}
            \label{fig:2_hexagon}
          \end{figure}

          By Lemma~\ref{lem:centralcycle2}, there exists an odd $M$-alternating cycle containing the edge $(o, a)$.
          Let $G_1$ be such an $M$-alternating cycle with minimum length.  
          By symmetry, we may assume that 
          $G_1$ has either $(b, o)$, $(c, o)$, or $(d, o)$~(Figure~\ref{fig:2_hexagon_G1}).
          Moreover, by the minimality of $G_1$, $G_1$ does not have $e$ or $f$.
          In each case, the path $P_1$ with vertices $\{o,e,f,a\}$ is of length $3$ and $G_1+P_1$ is a central subgraph.
          Therefore, the structure~(ii) exists in $G$.
          
          Thus, when $|D\cap M|=2$, $G$ has either structure~(i) or~(ii), which implies the existence of an admissible ear decomposition by Lemma~\ref{lem:centralcycle1}.
          
          \begin{figure}
            \centering
            \includegraphics[keepaspectratio,width=0.9\textwidth]
            {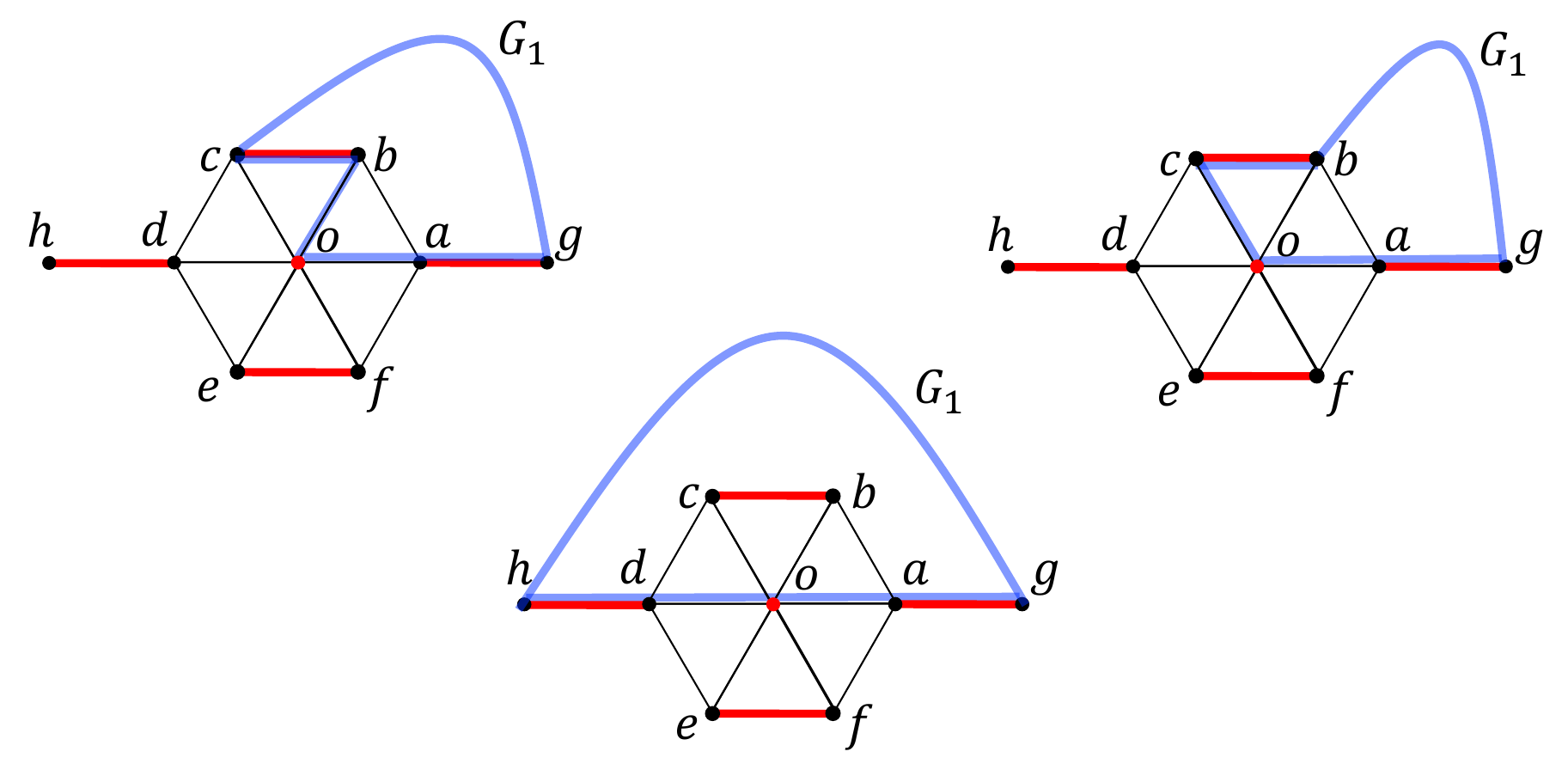}
            \caption{Case~(b) when $|D\cap M|=2$: The $3$ possibilities of $G_1$ through the edge $(o, a)$.}
            \label{fig:2_hexagon_G1}
          \end{figure}

\medskip
\noindent
\textbf{(c) $|D\cap M|=1$.}

          By symmetry, we may assume that $D\cap M=\{(b,c)\}$ as in Figure~\ref{fig:1_hexagon}~(left).
          Let us consider the edges of $M$ covering $a$ and $d$.
          Since we may assume that it has no central pentagons, there are $3$ possible configurations on the edges of $M$ as in Figure~\ref{fig:1_hexagon}~(middle).
          
          \begin{figure}
            \centering
            \includegraphics[keepaspectratio,width=1.0\textwidth]
            {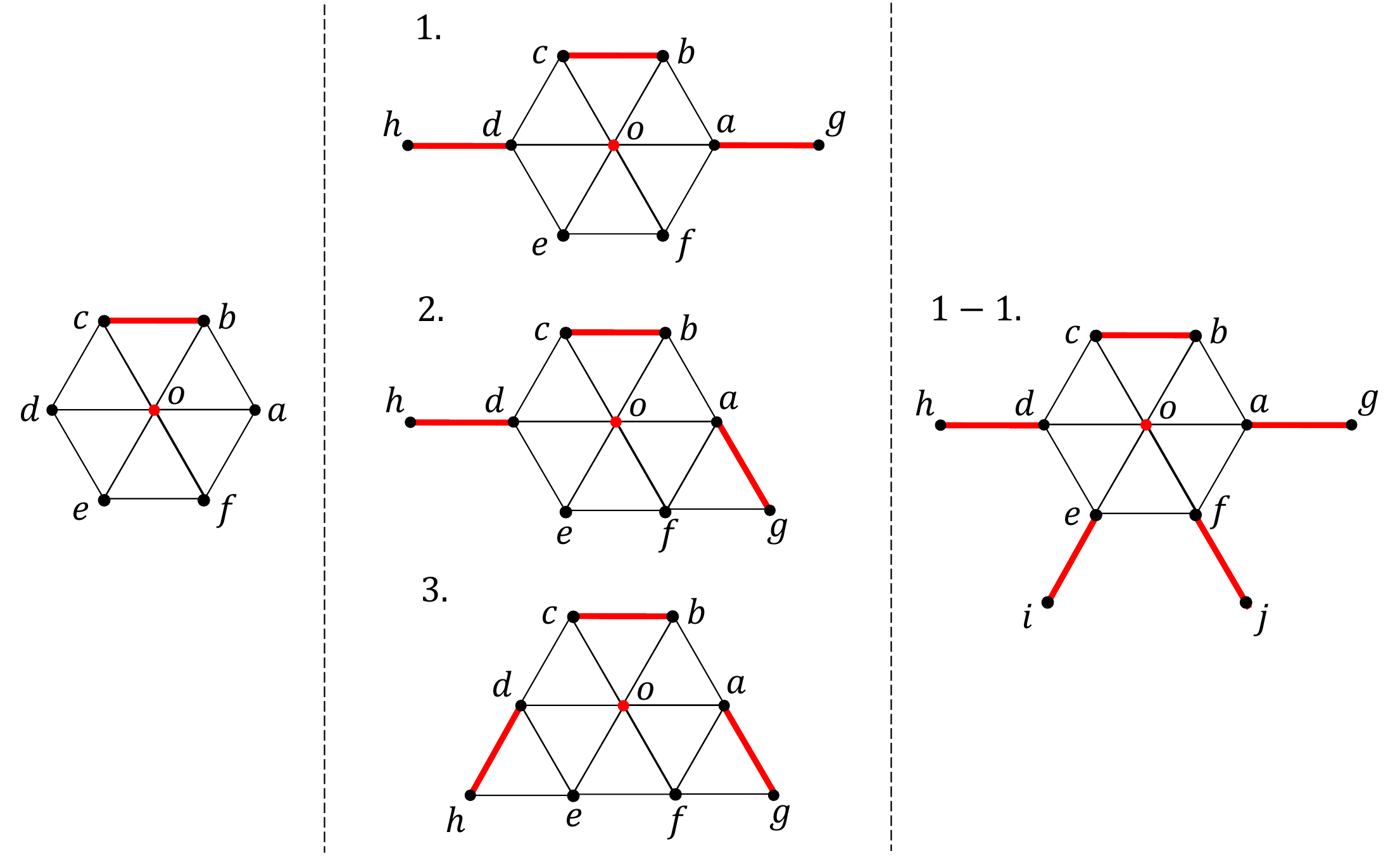}
            \caption{Case~(c) when $|D\cap M|=1$~(left).
            The $3$ possible patterns of matching edges covering $a$ and $d$~(middle).
            Assuming we have no central pentagon, it suffices to consider the right figure.}
            \label{fig:1_hexagon}
          \end{figure}

          We next consider the edges of $M$ covering $e$ and $f$.
          Then we can see that Patterns 2 and 3 in Figure~\ref{fig:1_hexagon}~(middle) have always central pentagons.
          Thus the structure~(i) exists.
          It remains to discuss Pattern 1 such that the vertices $e$ and $f$ have matching edges as in Figure~\ref{fig:1_hexagon}~(right).          
          By Lemma~\ref{lem:centralcycle2}, there exists an odd $M$-alternating cycle containing the edge $(o,a)$.
          Let $P$ be the path from $a$ to $o$ consisting of vertices $\{o,c,b,a\}$.

          Suppose that there exists an odd $M$-alternating cycle containing the edge $(o,a)$ such that it does not intersect with the path $P$.
          Let $G_1$ be such a cycle with minimum length.
          If the cycle $G_1$ has the edge $(d, o)$, then 
          $G_1$ with the path $P$ implies that the structure~(ii) exists.
          Similar arguments can be applied when the cycle $G_1$ has the edge $(e,o)$ or $(f,o)$.
          Thus we may assume that any odd $M$-alternating cycle containing the edges $(o,a)$ intersects with the path $P$.
          
          Suppose that there exists an odd $M$-alternating cycle containing the edges $(o,a)$ and $(c,o)$~(Figure~\ref{fig:C_bcoa}).
          Let $C$ be such an cycle with minimum length.
          The cycle $C$ has the edges $(b,c)$ and $(c,o)$.
          Let $Q$ be the subpath of $C$ between $a$ and $b$ consisting of vertices $\{b,c,o,a\}$, and $Q'$ be the subpath obtained from $C$ by removing the edges of $Q$.
          Also, let $M'=M\setminus\{(b,c)\}\cup\{(o,c)\}$.
          Then the cycle consisting of $Q'$ and the edge $(a,b)$ is an odd $M'$-alternating cycle $C'$, and $Q$ is an $M'$-alternating path such that the end vertices are adjacent on the cycle $C'$.
          The pair $C'$ and $Q$ yield the structure~(ii).
          
          \begin{figure}
            \centering
            \includegraphics[keepaspectratio,width=0.35\textwidth]
            {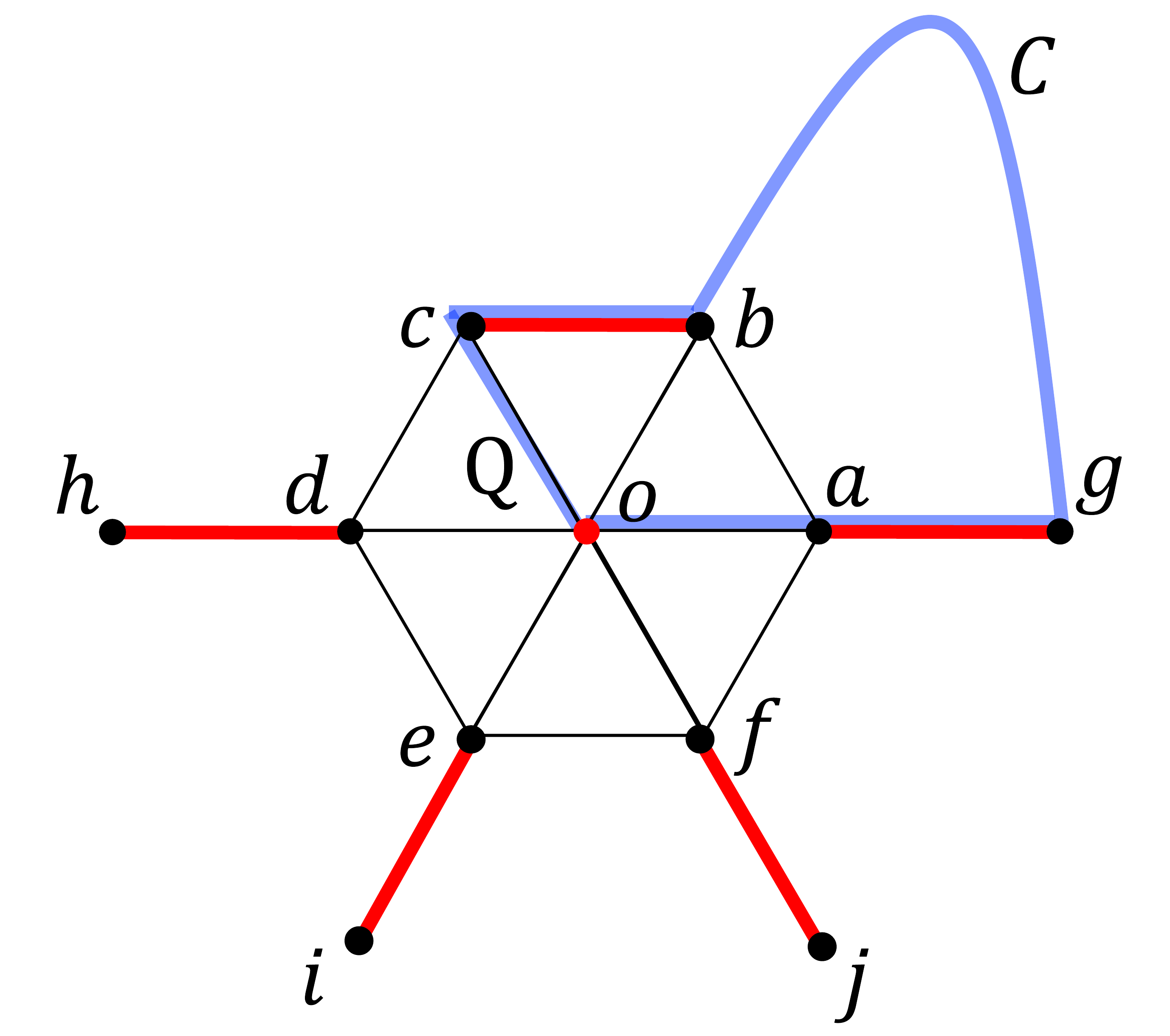}
            \caption{Case~(c) when the $M$-alternating cycle contains the edges $(a,o)$ and $(c,o)$.}
            \label{fig:C_bcoa}
          \end{figure}

          \begin{figure}
            \centering
            \includegraphics[keepaspectratio,width=0.35\textwidth]
            {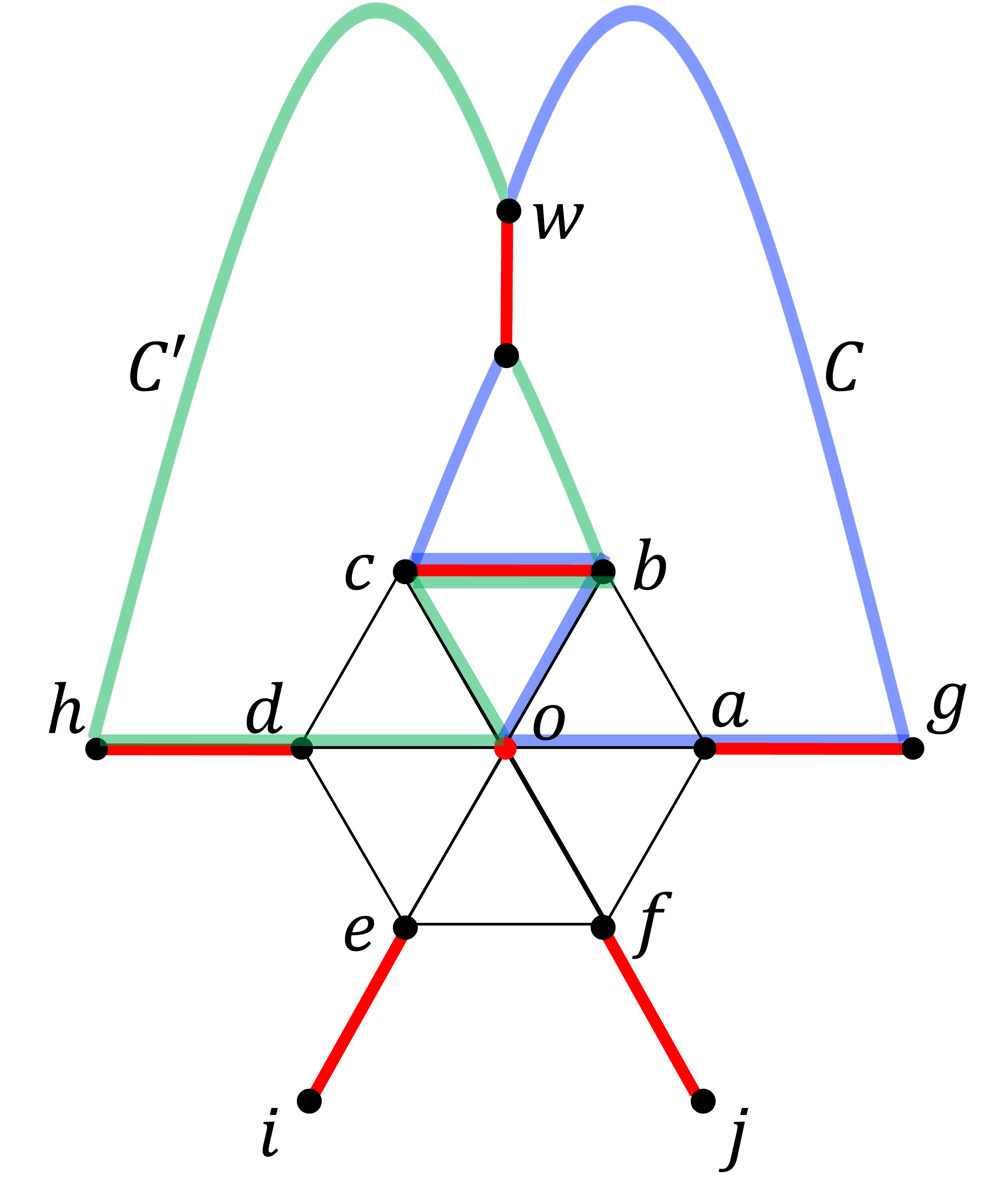}
            \caption{Case~(c) when the $M$-alternating cycle contains the edges $(a,o)$ and $(b,o)$.}
            \label{fig:1_hexagon_intersect}
          \end{figure}

          Therefore, we may suppose that any odd $M$-alternating cycle $C$ containing the edge $(o,a)$ has the edge $(b,o)$.
          This cycle $C$ contains the edge $(c,b)$~   (Figure~\ref{fig:1_hexagon_intersect}).
          Let $C'$ be an odd $M$-alternating cycle containing the edge $(o,d)$.
          Since we can apply the symmetrical argument to the vertex $d$, we may assume that $C'$ contains edges $(b,c$) and $(c,o)$.
          Since $C$ and $C'$ intersect by planarity and both are $M$-alternating cycles, 
          $C$ and $C'$ share at least one vertex.
          Let $w$ be the first vertex such that $C$ and $C'$ share with when traversing $C$ from $o$ through $a$.
          Let $Q$ be the cycle consisting of 
          the subpath of $C$ from $o$ to $w$ through $a$ and the subpath of $C'$ from $w$ to $o$ through $b$.
          Then $Q$ is an odd $M$-alternating cycle containing edges $(o,a)$ and $(c,o)$, which contradicts the assumption of this paragraph.
          
Thus, when $|D\cap M|=1$, $G$ has either structure~(i) or~(ii), which implies the existence of an admissible ear decomposition by Lemma~\ref{lem:centralcycle1}.

\medskip
\noindent
\textbf{(d) $D\cap M=\emptyset$.}

          The edges of $M$ adjacent to the vertices $a,b,c,d,e,f$ are placed as in the following $2$ patterns in Figure~\ref{fig:0_hexagon}.
          Let $g,h,i,j,k,l$ be the other end vertices of these edges of $M$.
          
          \begin{figure}
            \centering
            \includegraphics[keepaspectratio,width=0.6\textwidth]
            {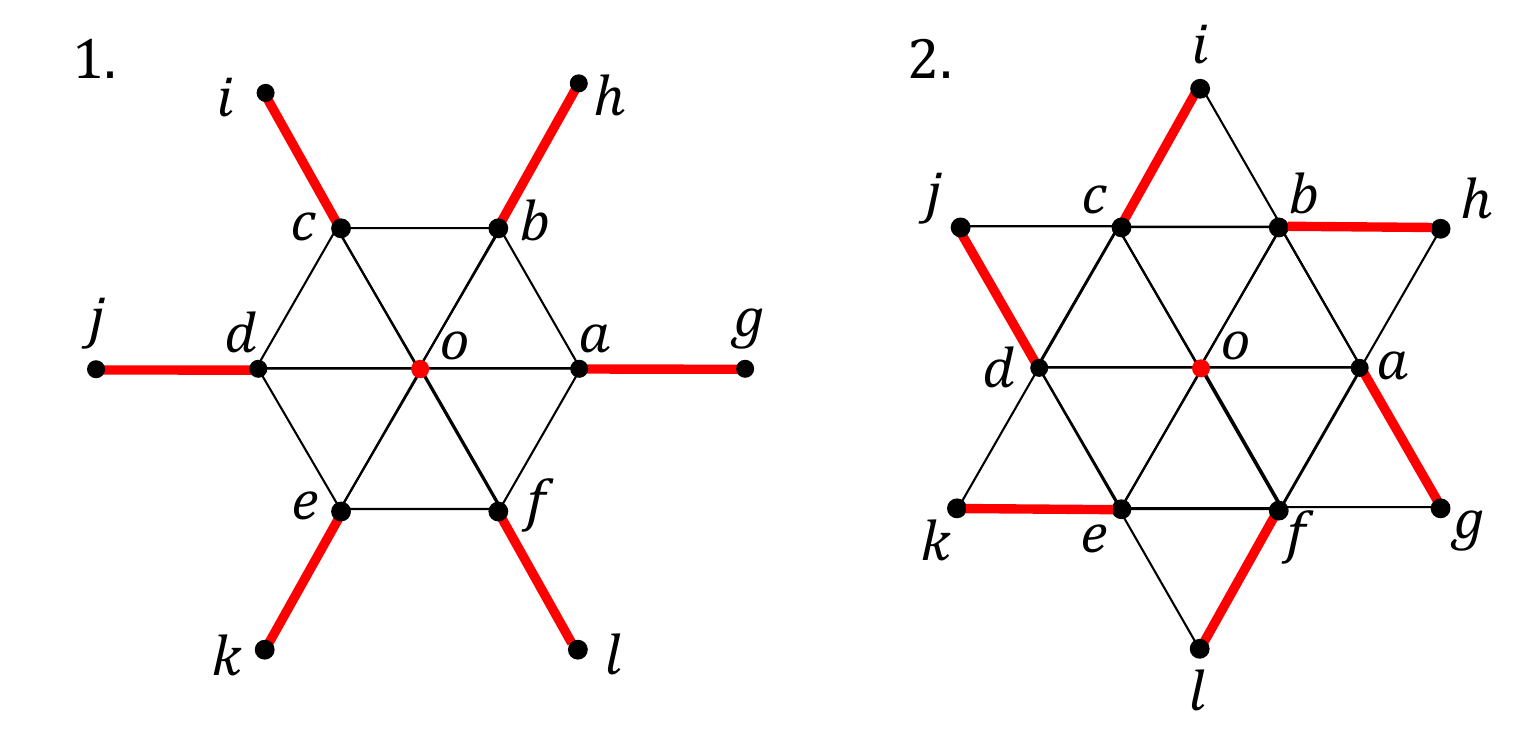}
            \caption{Case~(d) when $D\cap M=\emptyset$.}
            \label{fig:0_hexagon}
          \end{figure}

          We first show the following claim.
          
          \begin{claim}\label{clm:lastcase}
          There exist the adjacent vertices $u, v$ in $\{a, b,c,d,e,f\}$ such that $G$ has an odd $M$-alternating cycle containing the edges $(o,u)$ and $(o,v)$.    
          \end{claim}
          \begin{proof}
          By Lemma~\ref{lem:centralcycle2}, there exists an odd $M$-alternating cycle $C$ containing the edge $(o,a)$.
          By symmetry, we may suppose that $C$ has either $(b,o)$, $(c,o)$, or $(d,o)$.
          The claim holds if $C$ contains edge $(b,o)$.
          
            Suppose that there exists an odd $M$-alternating cycle $C$ containing the edges $(o,a)$ and $(c,o)$~(Figure~\ref{fig:0_hexagon_last}~(left)).
            Let $C'$ be an odd $M$-alternating cycle containing the edge $(o,b)$, which exists by Lemma~\ref{lem:centralcycle2}.
            Then, by planarity, $C$ intersects with $C'$.
            Let $w$ be the first common vertex in $C$ and $C'$ when traversing $C'$ from $o$ through $b$, and $Q$ be the subpath of $C'$ from $o$ to $w$ through $b$.
            Then the union of $C$ and $Q$ forms $2$ cycles, one through $(o, a)$ and $(o, b)$ and the other through $(o, b)$ and $(o, c)$.
            We see that one of them is an $M$-alternating cycle of odd length.
            Hence there exists an odd $M$-alternating cycle that satisfies  the condition of the claim, by taking either $\{u, v\}=\{a, b\}$ or $\{u, v\}=\{b, c\}$.
            
            Therefore, we may suppose that any odd $M$-alternating cycle containing the edge $(o,a)$ has the edge $(d,o)$.
            Moreover, by applying the symmetrical argument to the vertex $d$,
            we may suppose that any odd $M$-alternating cycle containing the edge $(o,d)$ has the edge $(a,o)$.
            Let $C'$ be an odd $M$-alternating cycle containing the edge $(o,b)$~(Figure~\ref{fig:0_hexagon_last}~(right)).
            First assume that $C'$ intersects with $C$.
            Then, similarly to the previous paragraph, the union of $C$ and $C'$ forms $2$ cycles, and one of them is $M$-alternating.
            If the cycle through edges $(o,a)$ and $(o,b)$ is $M$-alternating, then 
            the claim holds.
            Otherwise, if the cycle through edges $(o,b)$ and $(o,d)$ is $M$-alternating, it contradicts the assumption of this paragraph.
            Therefore, we may suppose that $C'$ does not intersect with $C$.
            Then $C'$ contains the edge $(c,o)$, and $C'$ satisfies the condition of the claim.
          \end{proof}

          \begin{figure}
            \centering
            \includegraphics[keepaspectratio,width=0.7\textwidth]
            {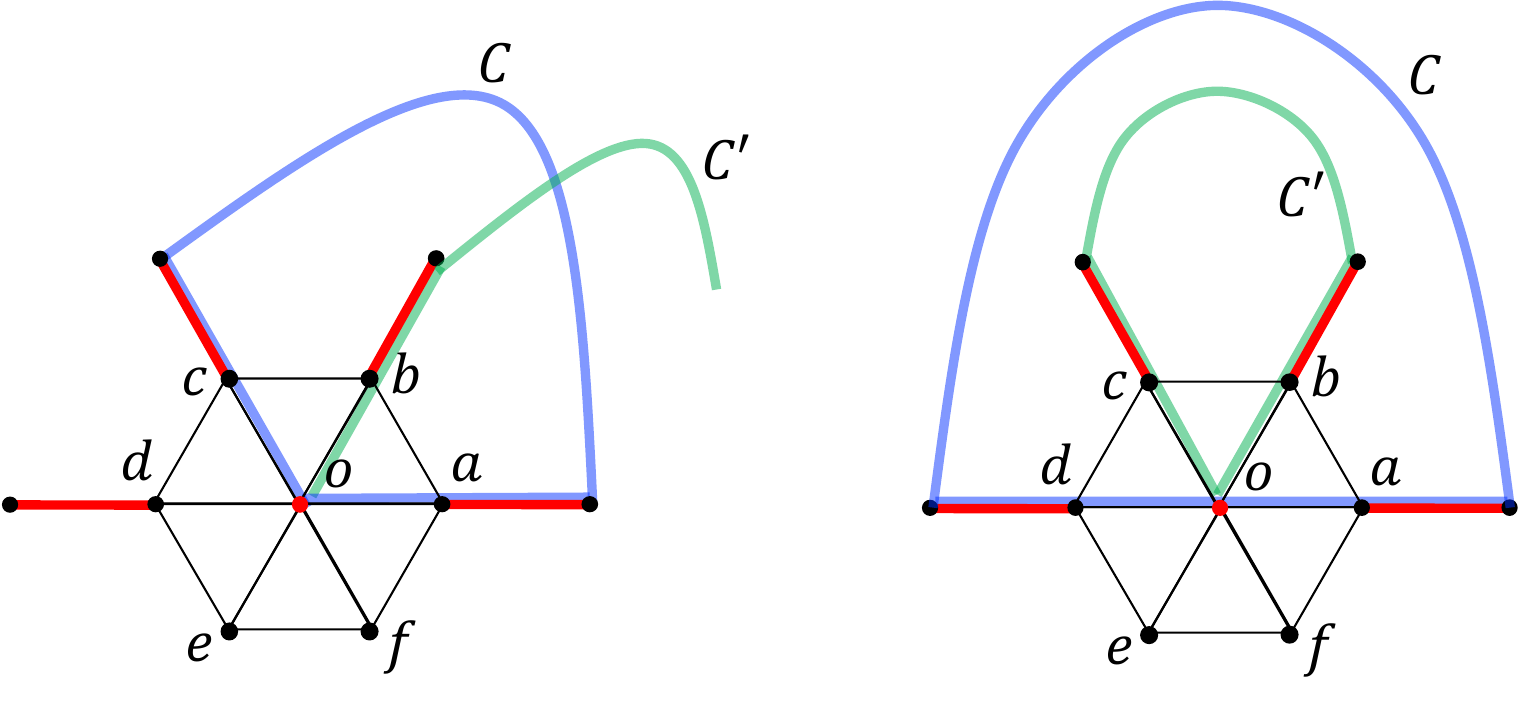}
            \caption{Proof of Claim~\ref{clm:lastcase}.}
            \label{fig:0_hexagon_last}
          \end{figure}

            Let $C$ be an $M$-alternating cycle found in Claim~\ref{clm:lastcase}.
            By symmetry, we may assume that $C$ has the edges $(o,a)$ and $(o,b)$.
            Let $C'=C\setminus\{(o,a), (o,b)\}\cup\{(a,b)\}$.
            Then $C'$ is an even $M$-alternating cycle.
            Define a nearly perfect matching $M'$ as $M'=M\triangle E(C')$.
            Then we can reduce to Case~(c) when we have exactly one matching edge on the boundary cycle $D$.            
            Therefore, it follows from Case~(c) that there is an admissible ear decomposition.

    Therefore, in each case, we can find an admissible ear decomposition.
\end{proof}



\section{Reconfiguration on Locally-Connected Graphs}
\label{sec:locally_connect}

In this section, we consider a triangular grid graph which is locally-connected.
  A vertex $u$ in a graph $G$ is said to be \textit{locally-connected} if the subgraph $G[N(u)]$ is connected.
  A graph $G$ is called \textit{locally-connected} if every vertex is locally-connected.
It is observed that, if $G$ is locally-connected, then it is $2$-connected, since, for a cut vertex $u$, the subgraph $G[N(u)]$ is disconnected.

The following theorem says that a locally-connected triangular grid graph, except for the Star of David graph~(Figure~\ref{fig:davide}), is Hamiltonian.
We note that, since their proof is constructive, a Hamilton cycle can be found in polynomial time.

\begin{theorem}[Gordon, Orlovich, and Werner~\cite{ref:triangle}]
  \label{thm:hamiltonian}
  Let $G$ be a triangular grid graph. 
  If $G$ is locally-connected, but not isomorphic to the Star of David graph, then it has a Hamilton cycle.
\end{theorem}

It follows from the above theorem that a locally-connected triangular grid graph, except for the Star of David graph, is factor-critical, as it has an odd and proper ear decomposition starting from a Hamilton cycle such that all the ears are single edges.
On the other hand, the Star of David graph is not factor-critical, and hence it is not reconfigurable~(see also~\cite{ref:gourds}).

The main theorem of this section is the following.

\begin{theorem}
  \label{thm:locally_algorithm}
Let $G=(V, E)$ be a triangular grid graph with $2n+1$ vertices.
If $G$ is locally-connected, but not isomorphic to the Star of David graph,  
then $G$ is reconfigurable. 
Moreover, a reconfiguration sequence using $O(n^3)$ \textsf{slide} operations can be found in polynomial time.
\end{theorem}

\subsection{Proof Overview}

The proof for Theorem~\ref{thm:locally_algorithm} exploits a Hamilton cycle in $G$ to design a reconfiguration sequence.
We note that the proof for $2$-connected graphs with no holes by Hamersma et al.~\cite{ref:gourds} 
 also uses a Hamilton cycle.
Our proof refines their proof so that we can deal with holes.

Suppose that we are given two placements $p$ and $q$. 
The proposed algorithm to reconfigure $p$ to $q$ consists of the following three phases.

\begin{enumerate}
  \item Reconfigure $p$ to a placement aligned with a Hamilton cycle $H$, denoted by $p'$.
  \item  Reconfigure $p'$ to another placement $q'$ aligned with $H$.
  \item  Reconfigure $q'$ to the target placement $q$.
\end{enumerate}

In Phase~1, we first reconfigure  the initial placement $p$ to a placement aligned with the Hamilton cycle $H$, which is denoted by $p'$.
Applying the same procedure to the target placement $q$, we obtain a placement aligned with $H$, denoted by $q'$.
We then reconfigure $p'$ to $q'$ in Phase~2.
In Phase~3, the placement $q'$ can be reconfigured to the target placement $q$ by taking the inverse of Phase~1 operations for $q$.

It was shown in Hamersma et al.~\cite{ref:gourds} that we can reconfigure a placement to a placement aligned with a Hamilton cycle.
We recall that a locally-connected graph is $2$-connected.

\begin{theorem}[Hamersma et al.~\cite{ref:gourds}]
  \label{thm:align}  
    Let $G$ be a $2$-connected triangular grid graph with $2n+1$ vertices.
    Then we can reconfigure any placement to a placement aligned with a Hamilton cycle $H$, using $O(n^2)$ \textsf{slide} operations.
\end{theorem}  

Therefore, Phases~1 and 3 can be implemented in $O(n^2)$ \textsf{slide} operations.
Thus it suffices to implement Phase~2 to reconfigure any placement aligned with $H$ to another placement aligned with $H$.
This step is realized by the following theorem.

\begin{theorem}\label{thm:lc_main}
Let $G$ be a triangular grid graph with $2n+1$ vertices, which is locally-connected, but not isomorphic to the Star of David graph.
Let $H$ be a Hamilton cycle.
For a pair of placements $p$, $q$ aligned with $H$, 
we can reconfigure $p$ to $q$ in $O(n^3)$ \textsf{slide} operations. 
\end{theorem}

The proof of Theorem~\ref{thm:lc_main} adopts a similar strategy to that of Theorem~\ref{thm:factoralgorithm}, where we employs a Hamilton cycle instead of an ear decomposition.
We identify a small subgraph that can be used to reconfigure placements aligned with $H$.
Recall that a \textit{diamond} is a subgraph induced by two adjacent triangles.

\begin{figure}
    \centering
    \includegraphics[keepaspectratio,width=0.6\textwidth]
    {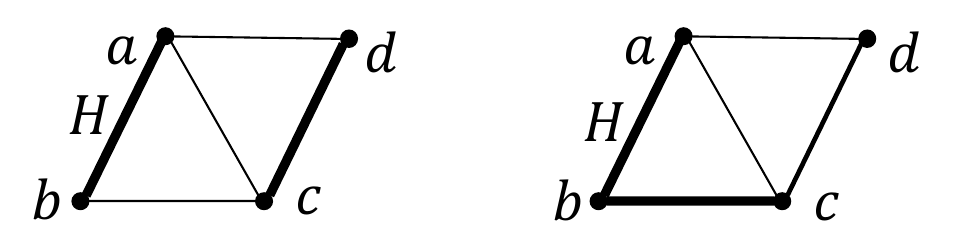}
    \caption{Diamonds and Hamilton cycles satisfying the condition in Lemma~\ref{lem:heikou_cycle}.}
    \label{fig:heikou_cycle}
  \end{figure}

\begin{restatable}{lemma}{heikoucycle}
  \label{lem:heikou_cycle}
  Let $G$ be a triangular grid graph with $2n+1$ vertices.
  Let $H$ be a Hamilton cycle of $G$.
  Suppose that $G$ has a diamond, whose vertices are $a, b, c, d$ aligned in the anti-clockwise order~(Figure~\ref{fig:heikou_cycle}), such that either 
  \begin{itemize}
  \item[\textup{(i)}] $H$ contains the edges $(a,b)$ and $(c,d)$, but does not contain $(a,c)$, or 
  \item[\textup{(ii)}] $H$ contains the edges $(a,b)$ and $(b,c)$.
  \end{itemize}
   Then we can reconfigure any placement aligned with $H$ to another placement aligned with $H$ in $O(n^3)$ \textsf{slide} operations.
\end{restatable}

See Section~\ref{sec:locallyconnectedreconf} for the proof.

We then show that such a diamond with a Hamilton cycle always exists if $G$ is a locally-connected triangular grid graph, but not isomorphic to the Star of David graph.
See Section~\ref{sec:lc_localstructure} for the details.
This, together with Lemma~\ref{lem:heikou_cycle}, shows Theorem~\ref{thm:lc_main}.

\begin{restatable}{lemma}{lclocal}
  \label{lem:LC_localstructure}
  Let $G$ be a triangular grid graph with $2n+1$ vertices.
If $G$ is a locally-connected triangular grid graph, but not isomorphic to the Star of David graph, then there exist a Hamilton cycle $H$ and a diamond, whose vertices are denoted by $a, b, c, d$~(Figure~\ref{fig:heikou_cycle}),  that satisfy either 
  \textup{(i)} or \textup{(ii)} in Lemma~\textup{\ref{lem:heikou_cycle}}.
\end{restatable}

This section is concluded with stating our results on the Gourds puzzle.

\begin{corollary}
  \label{cor:gourds2}
Let $B$ be a hexagonal grid such that the dual graph is locally-connected, but not isomorphic to the Star of David graph.
Then any two configurations of the same set of $n$ pieces on $B$ can be reconfigured to each other, using $O(n^3)$ moves.
\end{corollary}

\subsection{Reconfiguration using a Hamilton Cycle}\label{sec:locallyconnectedreconf}

This section is devoted to proving Lemma~\ref{lem:heikou_cycle}.
That is, we show that there exists a subgraph that can be used to reconfigure placements aligned with a Hamilton cycle.
The existence of such a subgraph will be discussed in Section~\ref{sec:lc_localstructure}.


By rotating along with a Hamilton cycle, we show that a triangular grid graph is reconfigurable if there exists a diamond satisfying parity conditions on a Hamilton cycle as below.

\begin{proposition}
  \label{prop:possible_cycle}
  Let $G$ be a triangular grid graph with $2n+1$ vertices.
  Let $H$ be a Hamilton cycle.
  Suppose that $G$ has a diamond, whose vertices are $a,b,c,d$ aligned in the anti-clockwise order, such that it satisfies the following $3$ conditions~(see Figure~\ref{fig:possible_cycle}):
  \begin{itemize}
    \item $H$ contains the edge $(a,b)$, and does not contain the edge $(a,c)$.
    \item The subpath of $H$ from $d$ to $a$ not containing $b$, denoted by $P_1$, has even length.
    \item The subpath of $H$ from $b$ to $c$ not containing $a$, denoted by $P_2$,  has odd length.
  \end{itemize}
  Then we can reconfigure any placement aligned with $H$ to another placement aligned with $H$ in $O(n^3)$ \textsf{slide} operations.
\end{proposition}

  \begin{figure}
    \centering
    \includegraphics[keepaspectratio,width=0.35\textwidth]
    {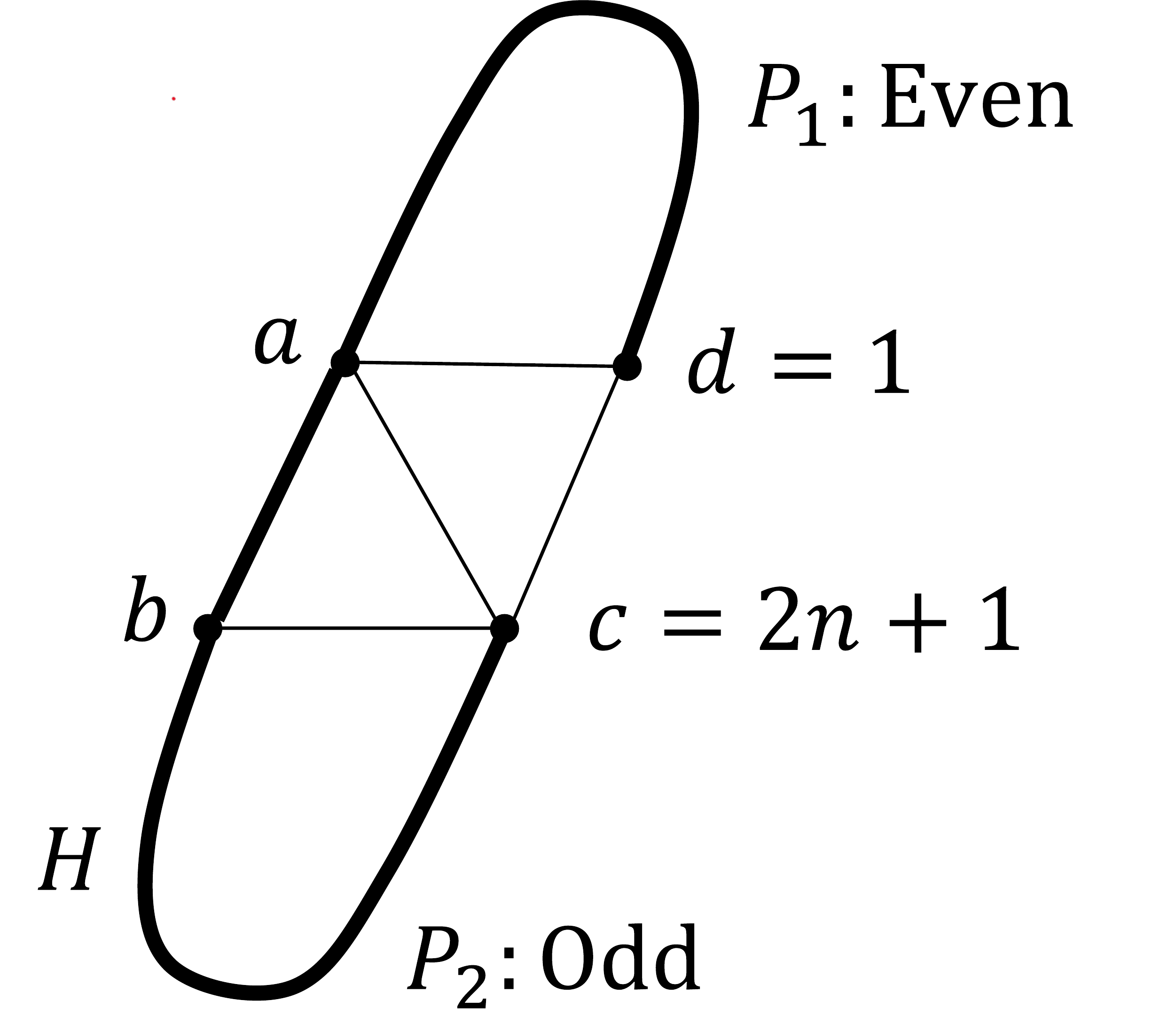}
    \caption{A diamond and a Hamilton cycle $H$~(thick cycle) satisfying the condition in  Proposition~\ref{prop:possible_cycle}}.
    \label{fig:possible_cycle}
  \end{figure}

Assuming the above proposition is true, we can prove Lemma~\ref{lem:heikou_cycle} as below.

\begin{proof}[Proof of Lemma~\ref{lem:heikou_cycle}]
 Suppose that $H$ contains the edges $(a,b)$ and $(c,d)$, but does not contain the edge $(a,c)$.
 Let $Q_1$ be the subpath of $H$ between $a$ and $d$ without $b$, and $Q_2$ be the subpath of $H$ between $b$ and $c$ without $a$.
 Since $H$ has odd length, we can see that one of $Q_1$ and $Q_2$ has odd length, and the other has even length.
Hence $H$ satisfies the condition of Proposition~\ref{prop:possible_cycle}.

 Next suppose that $H$ contains the edges $(a,b)$ and $(b,c)$.
 Note that $H$ goes through the vertex $d$.
 Let $Q_1$ be the subpath of $H$ between $a$ and $d$ without $b$, and $Q_2$ be the subpath of $H$ between $c$ and $d$ without $b$.
 Since $H$ has odd length, 
 one of $Q_1$ and $Q_2$ has odd length, and the other has even length.
 By symmetry, we may assume that $Q_1$ has even length, which is denoted by $P_1$.
 Then, letting $P_2$ be the path consisting of the edge $(b,c)$, 
 $H$ satisfies the condition of Proposition~\ref{prop:possible_cycle}.
 Thus Lemma~\ref{lem:heikou_cycle} holds.
\end{proof}

It remains to prove Proposition~\ref{prop:possible_cycle}.

\begin{proof}[Proof of Proposition~\ref{prop:possible_cycle}]
   We assume that the vertices $1,2,\dots, 2n+1$ of $G$ are aligned with the Hamilton cycle $H$ in the anti-clockwise order, where $d=1$ and $c=2n+1$ in the given diamond.
   We may also assume that $p(1), p(2), \dots, p(n)$ are aligned with $H$ in the anti-clockwise order and that $v_p=c=2n+1$, that is, $p(i) =(2i-1, 2i)$ for $i\in [n]$.
Note that we can make $c$ exposed in $O(n)$ \textsf{slide} operations along $H$ by  Observation~\ref{obs:rotation}.

  We will first prove that two consecutive pieces on $H$ can be swapped in $O(n)$ operations, by considering the following two cases according to the length of $P_1$.

\begin{claim}\label{clm:lc_swap1}
          Suppose that $P_1$ has length $2$.
          Then $p(1)$ and $p(2)$ can be swapped in a constant number of \textsf{slide} operations.
\end{claim}
   \begin{proof}
          In this case, $P_1$ has the unique inner vertex $v$~(Figure~\ref{fig:P1_short}).
          Thus $p(1)=(v,d)$ and $p(2)=(a,b)$.
          Then $p(1)$ and $p(2)$ can be swapped with the pentagon induced by the vertex set $\{a,b,c,d,v\}$ in a constant number of \textsf{slide} operations.
  Thus the claim holds.
   \end{proof}
          \begin{figure}
            \centering
            \includegraphics[keepaspectratio,width=0.3\textwidth]
            {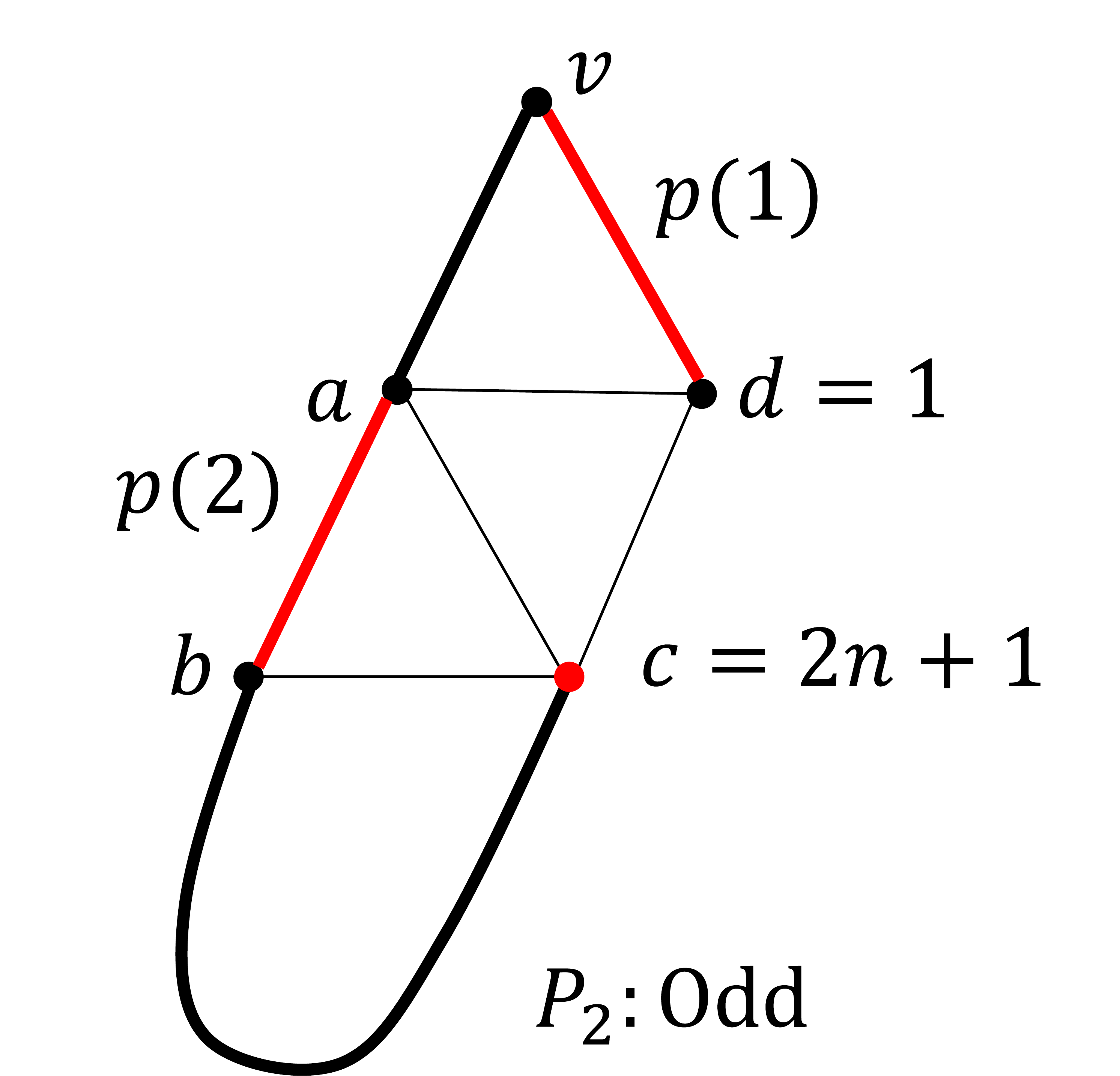}
            \caption{Proof of Claim~\ref{clm:lc_swap1} when $P_1$ has length $2$.}
            \label{fig:P1_short}
          \end{figure}
   
          Suppose that $P_1$ has length at least $4$.
          We denote by $v_1$ and $v_2$ the vertices next to $a$ on $P_1$, and by $v_3$ the vertex adjacent to $d$ on $P_1$.
          See Figure~\ref{fig:P1_long}.
          We assume that $p(j)=(v_2,v_1)$ and $p(j+1)=(a,b)$ for some $j\in [n-1]$.

          \begin{figure}
            \centering
            \includegraphics[keepaspectratio,width=0.9\textwidth]
            {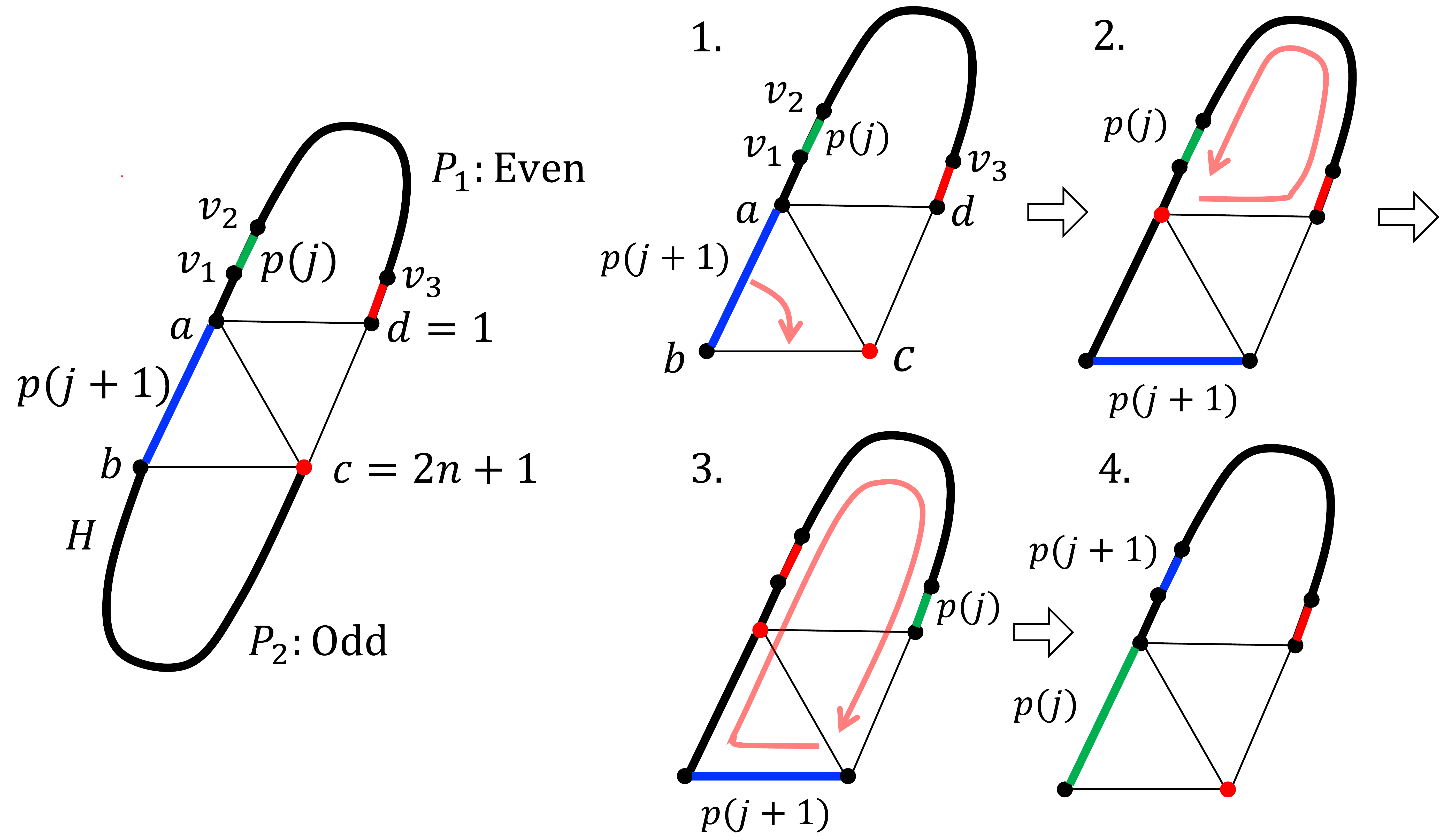}
            \caption{Proof of Claim~\ref{clm:lc_swap2} when $P_1$ has length at least $4$~(left). Reconfiguration steps to swap $p(j)$ and $p(j+1)$~(right).}
            \label{fig:P1_long}
          \end{figure}

\begin{claim}\label{clm:lc_swap2}
          Suppose that $P_1$ has length at least $4$.
          Then $p(j)$ and $p(j+1)$ can be swapped in $O(n)$ \textsf{slide} operations.
\end{claim}
   \begin{proof}
          We swap $p(j)$ and $p(j+1)$ as follows~(Figure~\ref{fig:P1_long}~(right)).
          Initially, we set $\tilde{p}=p$.

             \begin{enumerate}
            \item Slide $\tilde{p}(j+1)$ so that $\tilde{p}(j+1)=(b,c)$ to make the vertex $a$ exposed. 
            \item We rotate the current placement $\tilde{p}$ along the cycle consisting of $P_1$ with edge $(a,d)$ in the anti-clockwise order so that $\tilde{p}(j)=(d,v_3)$ and the exposed vertex is $a$.
            \item We rotate the current placement $\tilde{p}$ along the cycle consisting of the path $P_1$ with edges $(a,b), (b,c), (c,d)$ in the clockwise order so that $\tilde{p}(j+1)=(v_1,v_2)$ and $\tilde{p}(j)=(a,b)$ and that the exposed vertex is $c$.
           \end{enumerate}
          Then $\tilde{p}(j+1)=p(j)$ and $\tilde{p}(j)=p(j+1)$ hold, which means that $p(j)$ and $p(j+1)$ have been swapped without changing the placement of the other pieces.
          The number of \textsf{slide} operations is $O(n)$, as we rotate along the two cycles, each of which requires $O(n)$ \textsf{slide} operations by Observation~\ref{obs:rotation}.
  Thus the claim holds.          
   \end{proof}

  Let $q$ be a target placement aligned with $H$.
  We may assume that the vertex $c=2n+1$ is exposed in the placement $q$. 
  Then there exists a permutation $\sigma$ on $[n]$ such that $q(i)=p(\sigma(i))$ for any $i\in [n]$.
  The placement $p$ can be reconfigured to $q$ by swapping two consecutive pieces in $p$ so that the resulting ordering corresponds to $\sigma$.
  This can be implemented by the bubble sort algorithm.
  By Claims~\ref{clm:lc_swap1} and~\ref{clm:lc_swap2}, we can swap two consecutive pieces in some position by $O(n^2)$ \textsf{slide} operations, and then we can move the next two consecutive pieces in the position by $O(n)$ \textsf{slide} operations.
  Since the bubble sort algorithm requires to swap two consecutive pieces $O(n^2)$ times, the total number of operations is $O(n^3)$.
\end{proof}

\subsection{Proof of Lemma~\ref{lem:LC_localstructure}: Existence of a Local Structure}\label{sec:lc_localstructure}

In this section, we show Lemma~\ref{lem:LC_localstructure}, saying that there exist a diamond and a Hamilton cycle that satisfy the condition of Lemma~\ref{lem:heikou_cycle} if $G$ is a locally-connected triangular grid graph, but not isomorphic to the Star of David graph.

To show the lemma, we introduce the dual graph.

\begin{figure}
  \centering
  \includegraphics[keepaspectratio,width=0.8\textwidth]
  {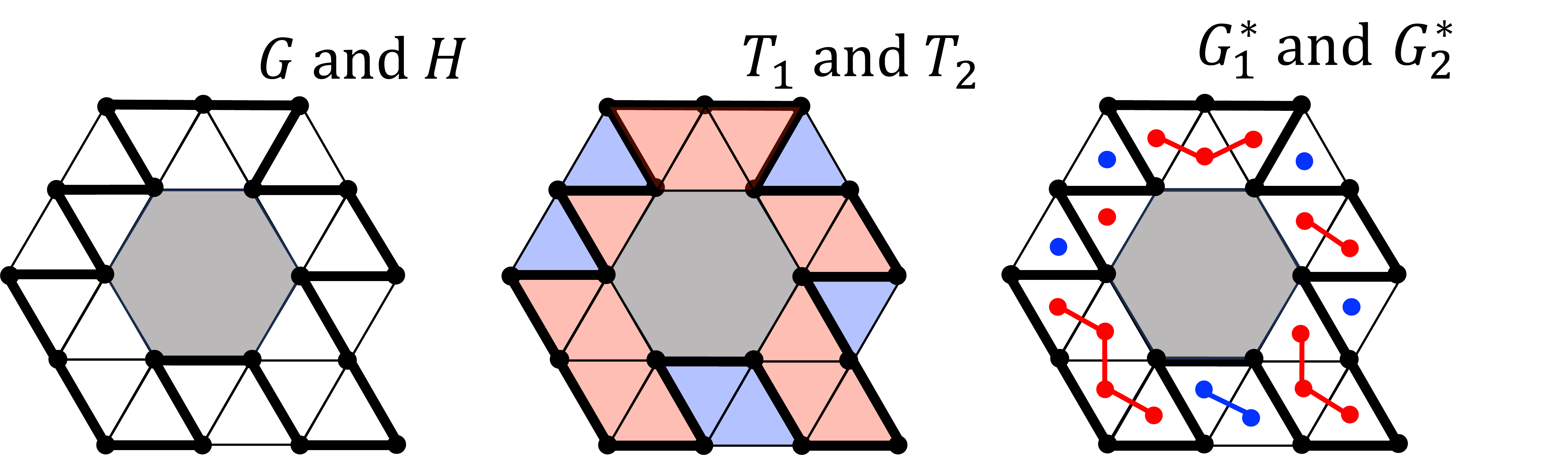}
  \caption{Examples of $G$ and $H$~(thick edges) and their dual graphs $G^\ast_1$~(blue) and $G^\ast_2$~(red).}
  \label{fig:G_T}
\end{figure}

Let $G$ be a triangular grid graph which is locally-connected, but not isomorphic to the Star of David graph.
We denote by $\B$ the family of the boundary cycles of the holes and the outer face.
We also denote the set of the inner edges of $G$ by $\Ein =E\setminus \bigcup_{C\in\B}E(C)$.
In other words, $\Ein$ is the set of edges that are not aligned on any boundary cycles in $\B$.

Since $G$ is planar, we can define the dual graph of $G$.
We here construct the dual graph $G^\ast$ only from the edge set $\Ein$.
That is, the vertex set of $G^\ast$ is the set of triangle faces in $G$, denoted by $T$.
Moreover, $G^\ast$ has an edge between two triangles $t$ and $t'$ if they are adjacent in $G$.
Thus the edge set of $G^\ast$ is identical to $\Ein$.

By Theorem~\ref{thm:hamiltonian}, $G$ has a Hamilton cycle $H$.
Since $H$ is a closed curve on the plane, the Hamilton cycle $H$ divides $G$ into two parts~(Figure~\ref{fig:G_T}~(left)).
Let $G^\ast_1$ and $G^\ast_2$ be the graphs obtained from the dual graph $G^\ast$ by removing the edges of $H\cap \Ein$ in $G^\ast$ such that $G^\ast_1$ corresponds to the part having the outer face.
Let $T_i$ be the set of triangles in $G^\ast_i$ for $i=1,2$.

We first present basic observations for the dual graphs.

\begin{observation}\label{obs:dual}
If an edge $e$ is shared with two triangles each from $G^\ast_1$ and $G^\ast_2$, then the edge $e$ is contained in $H$.
\end{observation}
\begin{proof}
The claim holds by definition of the dual graphs.
\end{proof}

\begin{lemma}
  \label{lem:3degree_tree}
  Both $G^\ast_1$ and $G^\ast_2$ are forests with maximum degree at most $3$. 
\end{lemma}

\begin{proof}
  Since each triangle in $T_1\cup T_2$ has at most $3$ adjacent triangles, each vertex of $G^\ast_1$ and $G^\ast_2$ is of degree at most $3$.
  We will show that $G^\ast_1$ is a forest.
  The case of $G^\ast_2$ can be proved in a similar way.
  Suppose to the contrary that $G^\ast_1$ has a cycle $C$.
  Then the cycle $C$ corresponds to a cut of $G$, which divides $G$ into two parts $G'_1$ and $G'_2$.
  Since $C$ has no edge of $H$, $H$ is included in either $G'_1$ or $G'_2$.
  Since $G'_1$ and $G'_2$ have vertices, this contradicts that $H$ is a Hamilton cycle.
  Thus $G^\ast_1$ has no cycles, which means that it is a forest.
\end{proof}

It follows from Lemma~\ref{lem:3degree_tree} that each connected component of $G^\ast_1$ and $G^\ast_2$ is a tree.
Moreover, the following lemma claims that there exists such a connected component with at least $2$ vertices.

  \begin{figure}
    \centering
    \includegraphics[keepaspectratio,width=0.8\textwidth]
    {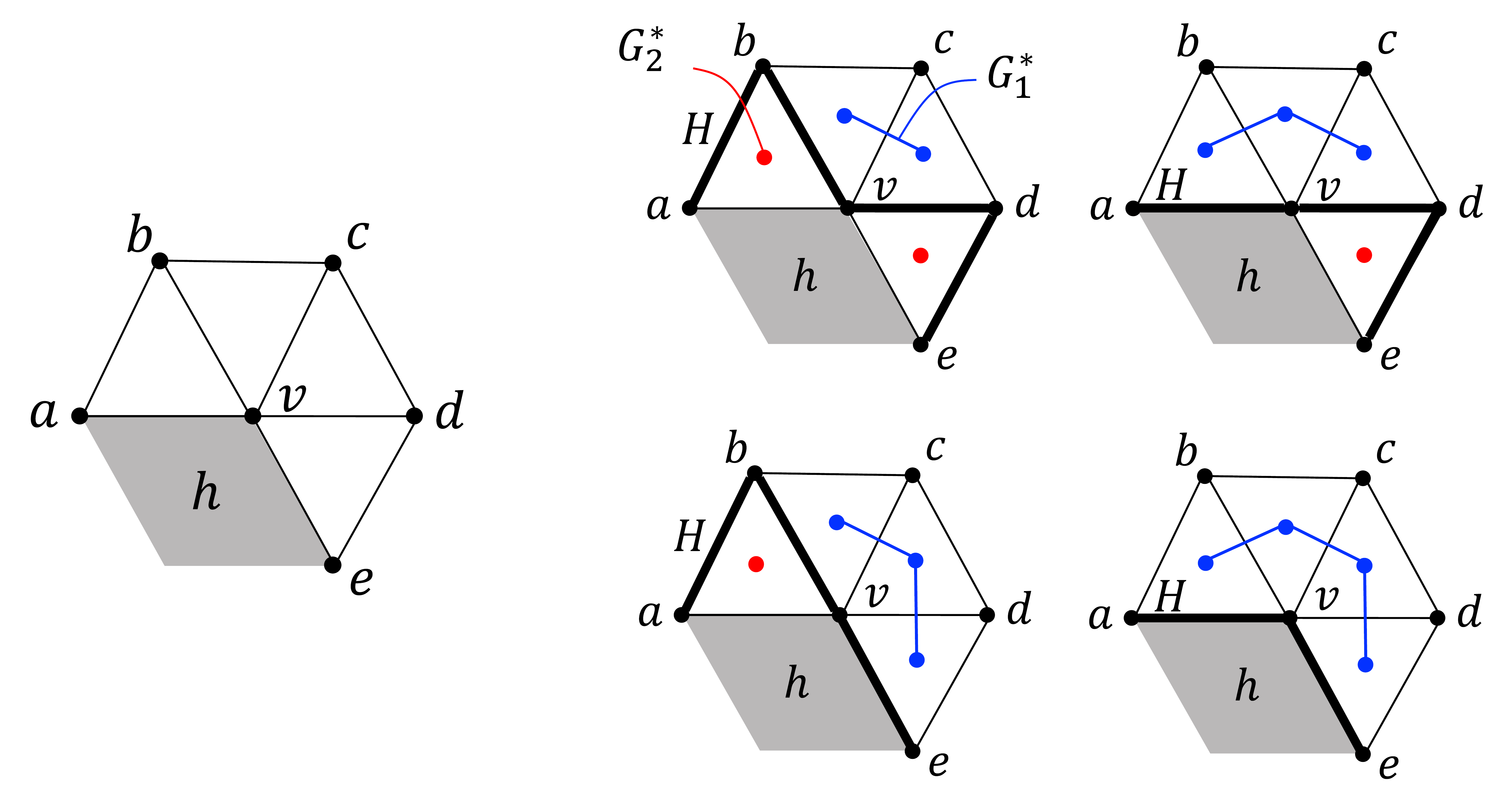}
    \caption{Proof of Lemma~\ref{lem:least_2vertice}. A vertex $v$ of a hole with interior angle 120$\tcdegree$~(left). The $4$ possibilities that $H$ goes through $v$, assuming that $G^\ast_2$ consists of singletons~(right).}
    \label{fig:120corner}
  \end{figure}

\begin{lemma}
  \label{lem:least_2vertice}
  Suppose that $G$ has at least $5$ vertices.
  Then at least one of $G^\ast_1$ and $G^\ast_2$ has a connected component with at least $2$ vertices.
\end{lemma}

\begin{proof}
  If $G^\ast_2$ has no holes, then $G^\ast_2$ is connected, and hence it is a tree by Lemma~\ref{lem:3degree_tree}.
  Since $H$ has at least $5$ vertices, $G^\ast_2$ has at least $2$ vertices.
  Thus we may assume that $G^\ast_2$ has a hole $h$.
  
  We will show that, if $G^\ast_2$ consists of singletons, then $G^\ast_1$ has a connected component with at least $2$ vertices.
  Since a triangular grid graph is an induced subgraph of the triangular lattice, the hole $h$ of $G$ has a vertex $v$ with interior angle $120\tcdegree$.
  Since $G$ is locally-connected, $v$ has degree $5$.
  We denote by $a,b,c,d,e$ the vertices in $N(v)$ surrounded by $v$ in the clockwise order~(Figure~\ref{fig:120corner}~(left)).
  Since the Hamilton cycle $H$ goes through $v$ and $G^\ast_2$ consists of singletons, the possible ordering of traversing the vertices of $N(v)$ along $H$ are $(a,b,v,d,e)$, $(a,v,d,e)$,
  $(a,b,v,e)$, and $(a,v,e)$~(Figure~\ref{fig:120corner}~(right)).
  In each case, we can see that the triangles induced by $\{b,c,v\}$ and $\{c,d,v\}$ are contained in $G^\ast_1$, and hence $G^\ast_1$ has a connected component with at least $2$ vertices.
  Thus the lemma holds.
  \end{proof}

We need the following technical lemma on the cycle surrounding a connected component in $G^\ast_1$ and $G^\ast_2$.

\begin{lemma}
  \label{lem:outside_hole}
  Let $G'$ be a connected component in $G^\ast_1$ or $G^\ast_2$, and $T'$ be the set of triangles corresponding to $G'$. 
  We denote by $C$ the boundary cycle of $T'$ in $G$.
  Then the following two statements hold.
  \begin{enumerate}
  \item[\textup{(1)}]   If an edge $e\in C$ is not in Hamilton cycle $H$, then $e$ is on some boundary cycle in $\mathcal{B}$.
  \item[\textup{(2)}]   For each consecutive two edges in $C$, at least one of them is contained in $H$.  
  \end{enumerate}
  \end{lemma}

\begin{figure}
  \centering
  \includegraphics[keepaspectratio,width=0.75\textwidth]
  {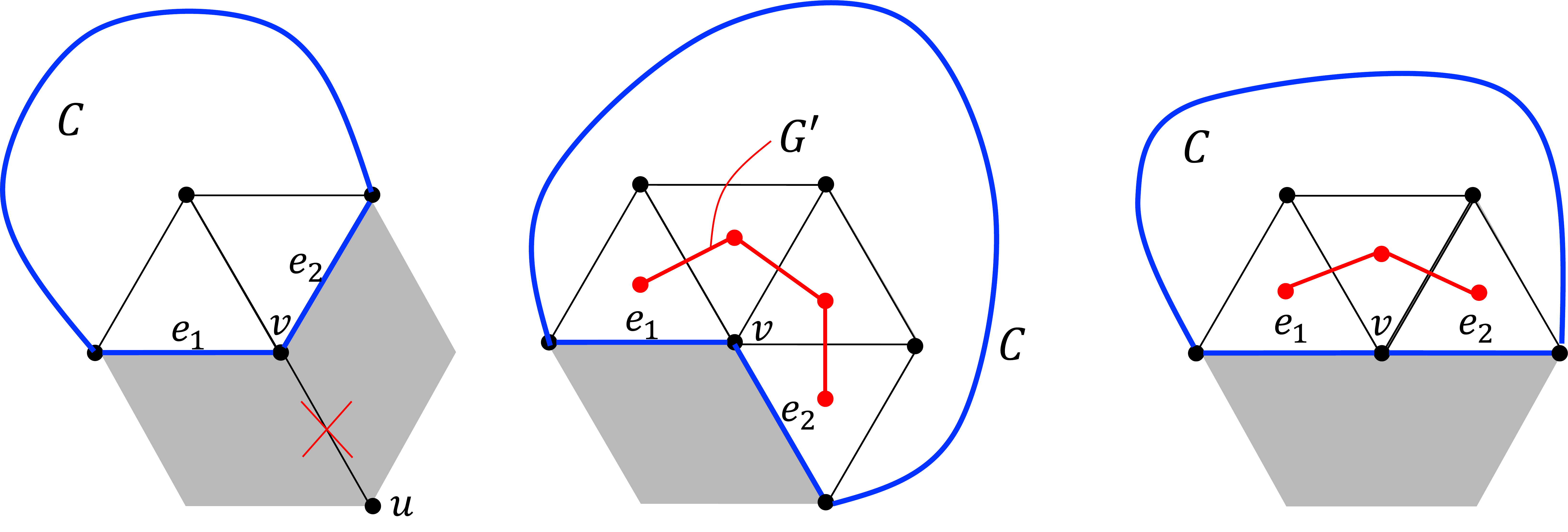}
  \caption{Proof of Lemma~\ref{lem:outside_hole}~(2).}
  \label{fig:outside_cycle}
\end{figure}

\begin{proof}
  (1)
  We will show that edge $e$ does not belong to two distinct triangles.
  Suppose $e$ is contained in two triangles. 
  Since $e$ is an edge on the boundary cycle of $T'$, one triangle belongs to $T_1$ and the other belongs to $T_2$.
  This contradicts Observation~\ref{obs:dual}.
  Thus any edge $e\in C$ cannot belong to two triangles, implying that $e$ must be on some cycle in $\B$.
  
  (2)
  Let $e_1,e_2$ be two consecutive edges on $C$.
  We denote by $v$ the common end vertex of $e_1$ and $e_2$.
  Suppose to the contrary that none of $e_1, e_2$ belongs to $H$.
  By~(1), $e_1$ and $e_2$ are in some cycles in $\B$.
  Since $v$ is locally-connected, $e_1$ and $e_2$ are on the same cycle in $\B$, and the triangles containing $v$ are all included in $T'$~(see e.g., Figure~\ref{fig:outside_cycle}).  
  Then all the edges incident to $v$ are not contained in $H$, which contradicts that $H$ goes through $v$.
  Thus at least one of $e_1$ and $e_2$ belongs to $H$.
\end{proof}

We also use the following simple observation given by Hamersma at al.~\cite{ref:gourds}.

\begin{lemma}
  \label{lem:substructure}  
  Let $G'$ be a connected component with at least $3$ vertices in $G^\ast_1$ or $G^\ast_2$.
  Then $G'$ satisfies one of the following conditions.
  \begin{enumerate}
    \item[\textup{(a)}]  $G'$ has a leaf adjacent to a vertex of degree $2$.
    \item[\textup{(b)}]  $G'$ has a vertex of degree $3$ adjacent to at least $2$ leaves.
  \end{enumerate}
  \end{lemma}

We are now ready to prove the main lemma~(Lemma~\ref{lem:LC_localstructure}) restated as below.

\lclocal*

\begin{proof}
  By Lemma~\ref{lem:least_2vertice}, 
  there exists a tree, denoted by $G'$, with at least $2$ vertices in $G^\ast_1$ or $G^\ast_2$.
  Let $C$ be the boundary cycle of $G'$.

  Suppose that $G'$ has only $2$ vertices.
  Then $G'$ corresponds to a diamond, and $C$ has length $4$.
  Since $C$ has edges not contained in $H$, Lemma~\ref{lem:outside_hole}, with symmetry, implies that the diamond $G'$ satisfies one of the conditions of Lemma~\ref{lem:LC_localstructure}.

  Thus we may suppose that $G'$ has at least $3$ vertices.
By Lemma~\ref{lem:substructure}, $G'$ has either a leaf adjacent to a vertex of degree 2, or
a vertex of degree 3 adjacent to two leaves.

\medskip
\noindent
\textbf{(a)~When $G'$ has a leaf $v$ adjacent to a vertex $u$ of degree $2$.}
                  
          The triangles corresponding to $u$ and $v$ form a diamond.
          We denote the vertices of the diamond by $a, b, c, d$ as depicted in Figure~\ref{fig:1leaf}~(left).
          Since $C$ is the boundary cycle of $G'$, the cycle $C$ contains the edges $(a,b)$, $(b,c)$, and $(c,d)$.
          Let $C'=\{(a,b), (b,c), (c,d)\}$.
          By Lemma~\ref{lem:outside_hole}~(2) with symmetry, it suffices to consider the $5$ configurations of $H$ on $C'$ as depicted in Figure~\ref{fig:1leaf}~(right).
          The first $4$ configurations satisfy one of the conditions of Lemma~\ref{lem:LC_localstructure}.

 \begin{figure}
                    \centering
                    \includegraphics[keepaspectratio,width=0.8\textwidth]
                    {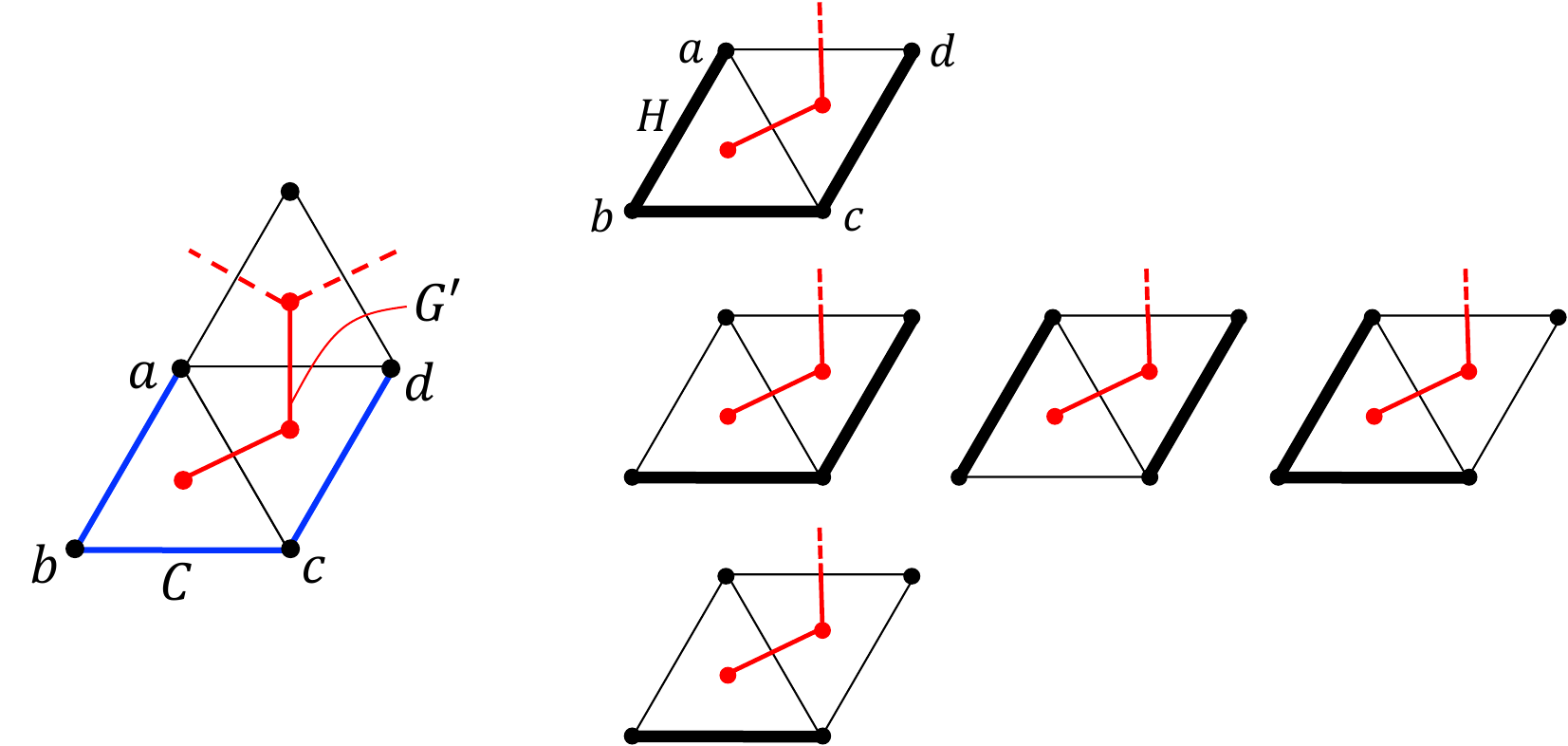}
                    \caption{A diamond of $G'$ in Case~(a) of the proof of Lemma~\ref{lem:LC_localstructure}~(left).
                    The $5$ possible patterns that $H$ goes through the $3$ edges $(a,b), (b,c), (c,d)$, where the thick edges are contained in $H$~(right).}
                    \label{fig:1leaf}
                  \end{figure}

          \begin{figure}
            \centering
            \includegraphics[keepaspectratio,width=0.7\textwidth]
            {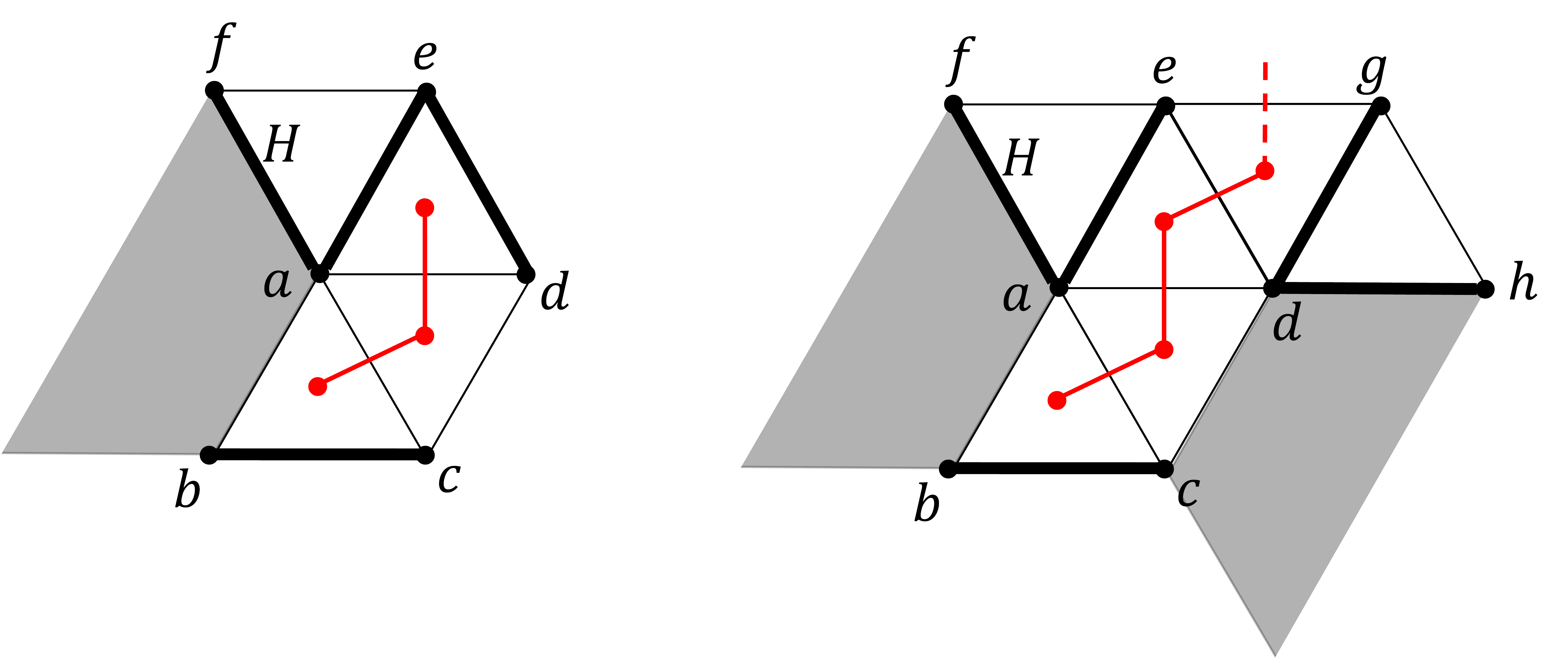}
            \caption{Case~(a) when $H$ contains only the edge $(b,c)$.}
            \label{fig:1leaf_hamilton}
          \end{figure}

     The remaining case is the last case, that is, when $H$ contains only the edge $(b,c)$, but not contain $(a,b)$ or $(c,d)$ on $C'$.
     In this case, Lemma~\ref{lem:outside_hole}~(1) implies that the edge $(a,b)$ is on some boundary cycle in $\B$.
     Hence $a$ is of degree at most $5$.
     Since $H$ contains the vertex $a$ without using edges $(a,b)$, $(a,c)$, or $(a,d)$, the vertex $a$ has $2$ more edges contained in $H$.
     We denote the end vertices of the $2$ edges by $e$ and $f$.
      Thus $H$ contains the edges $(a,e)$ and $(a,f)$~(Figure~\ref{fig:1leaf_hamilton}).

        If $H$ contains the edge $(d,e)$, then the diamond induced by $\{a, c, d, e\}$ satisfies the condition~(ii) of Lemma~\ref{lem:LC_localstructure}~(Figure~\ref{fig:1leaf_hamilton}~(left)).
          Thus we may assume that $H$ does not contain the edge $(d,e)$.
          Using the same argument applied for the vertex $d$, we can see that the vertex $d$ has $2$ more edges contained in $H$, denoted by $(d,g)$ and $(d,h)$~(Figure~\ref{fig:1leaf_hamilton}~(right)).
          In this case, the diamond induced by $\{a, d, g, e\}$ satisfies the condition~(i) of Lemma~\ref{lem:LC_localstructure}.

          Thus, if  $G'$ has a leaf $v$ adjacent to a vertex $u$ of degree $2$, there exists a diamond satisfying one of the conditions of Lemma~\ref{lem:LC_localstructure}.

\medskip
\noindent
\textbf{(b)~When $G'$ has a vertex $v$ of degree $3$ adjacent to two leaves $u$ and $w$. }

         The $3$ triangles corresponding to the vertices $v$, $u$, and $w$ form a pentagon.
         The vertices of the pentagon is denoted by $a, b, c, d, e$ as depicted in Figure~\ref{fig:2leaf}~(left).
         Then the boundary cycle $C$ contains the $4$ edges $(a,b), (b,c), (c,d)$, and $(d,e)$.
         Let $C'=\{(a,b), (b,c), (c,d), (d,e)\}$.
         By Lemma~\ref{lem:outside_hole}~(2) with symmetry,
         the possible configurations of $H$ on $C'$ are the $5$ configurations as depicted in Figure~\ref{fig:2leaf}~(right).
         For the first $3$ configurations, we find a diamond satisfying the condition~(ii) of Lemma~\ref{lem:LC_localstructure}.
         Thus the remaining cases are the latter $2$ configurations.

          \begin{figure}
            \centering
            \includegraphics[keepaspectratio,width=0.7\textwidth]
            {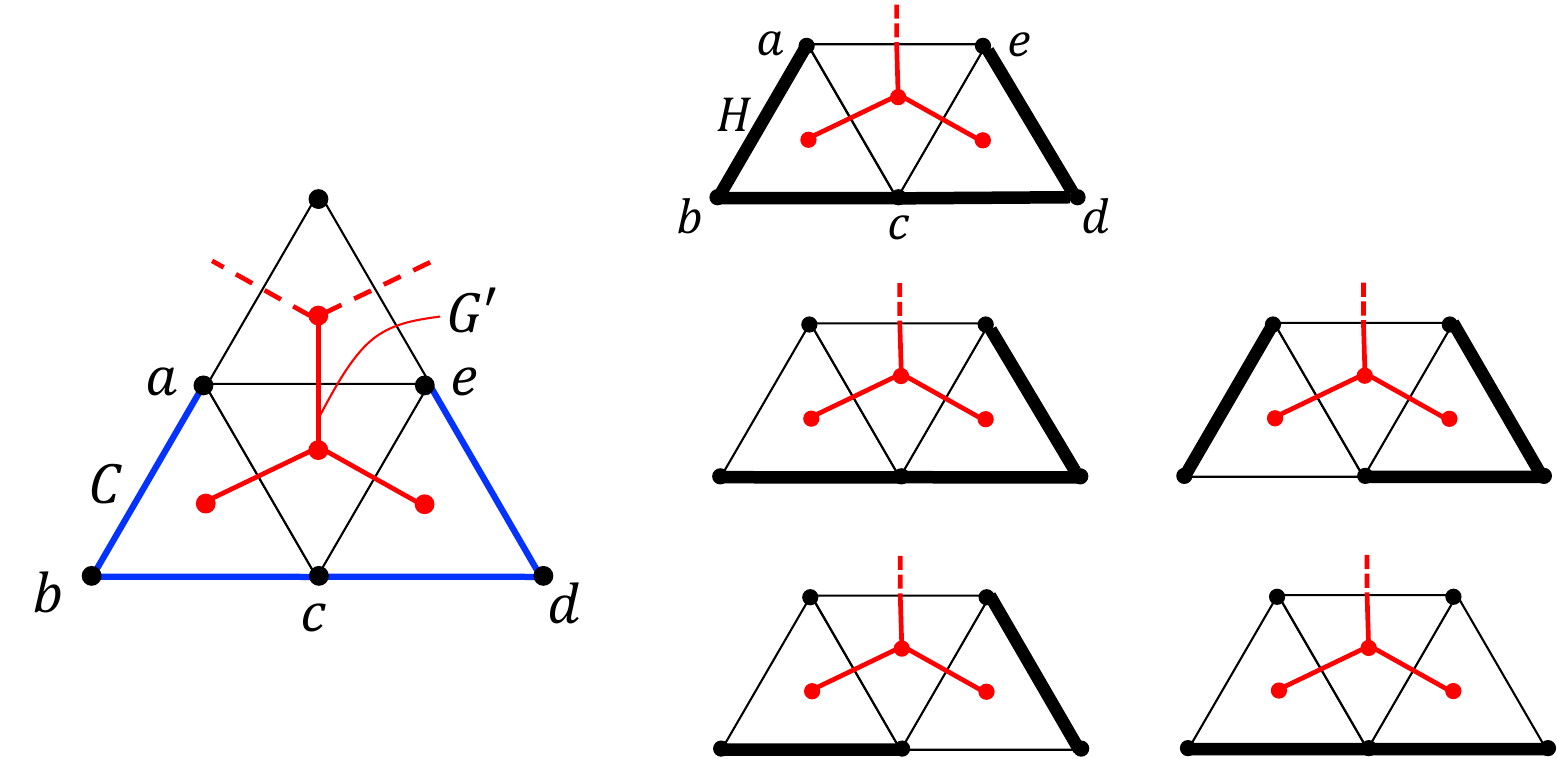}
            \caption{A pentagon in Case~(b)~(left).
            The $5$ possible patterns that $H$ goes through the $4$ edges $(a,b), (b,c), (c,d), (d,e)$~(right).}
            \label{fig:2leaf}
          \end{figure}

          \begin{figure}
            \centering
            \includegraphics[keepaspectratio,width=0.7\textwidth]
            {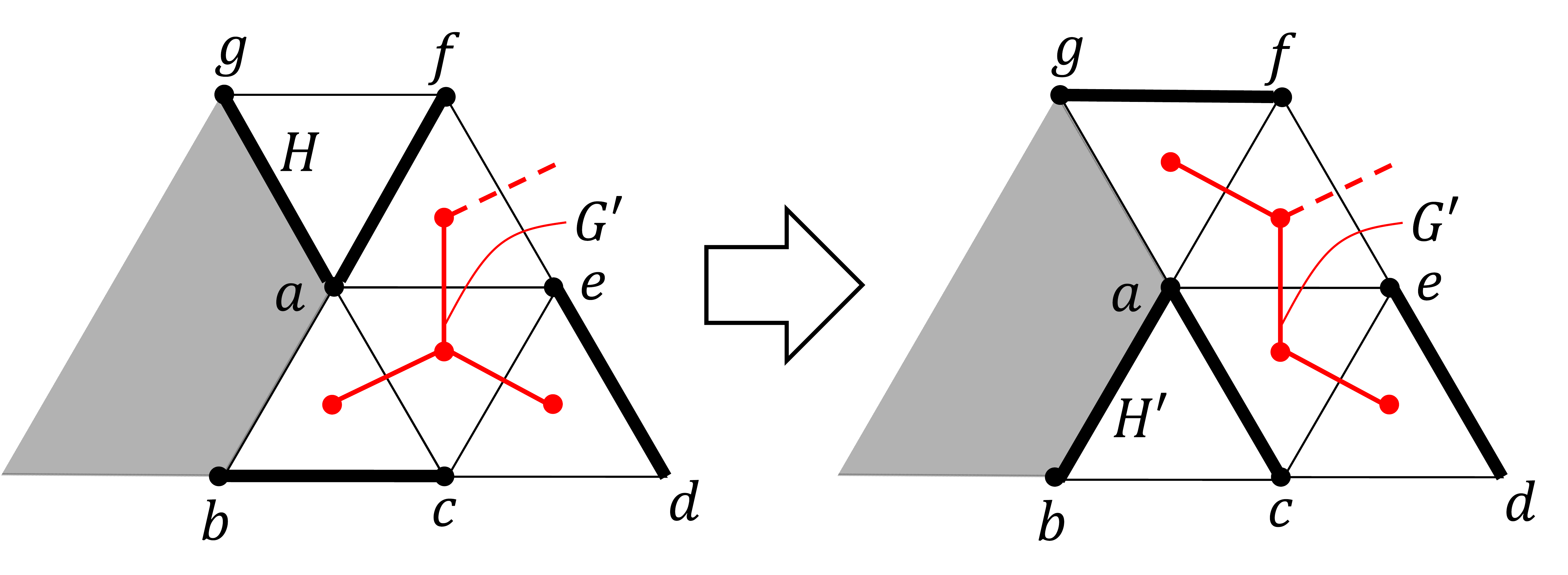}
            \caption{Case~(b): How to modify $H$ to $H'$ when $H$ contains only the edges $(b,c)$ and $(d,e)$.}
            \label{fig:2leaf_hamilton_1}
          \end{figure}

          First consider the case when $H$ contains only the edges $(b,c)$ and $(d,e)$ on $C'$.
          Then, the edge $(a,b)$ is on some boundary cycle in $\B$ by Lemma~\ref{lem:outside_hole}~(1).
          Since $H$ goes through the vertex $a$ without using $(a,b), (a,c)$ or $(a,e)$, the vertex $a$ has two more neighbors, denoted by $f, g$, and $H$ contains the edges $(a,f)$ and $(a,g)$~(Figure~\ref{fig:2leaf_hamilton_1}~(left)).
          We define another Hamilton cycle $H'$ from $H$ by removing the edges $(a,f), (a,g), (b,c)$ and appending the edges $(a,b), (a,c), (f,g)$~(Figure~\ref{fig:2leaf_hamilton_1}~(right)).
          Then the obtained Hamilton cycle $H'$ with a diamond induced by $\{a,c,d,e\}$ satisfies the condition~(i) of Lemma~\ref{lem:LC_localstructure}.

          \begin{figure}
            \centering
            \includegraphics[keepaspectratio,width=0.35\textwidth]
            {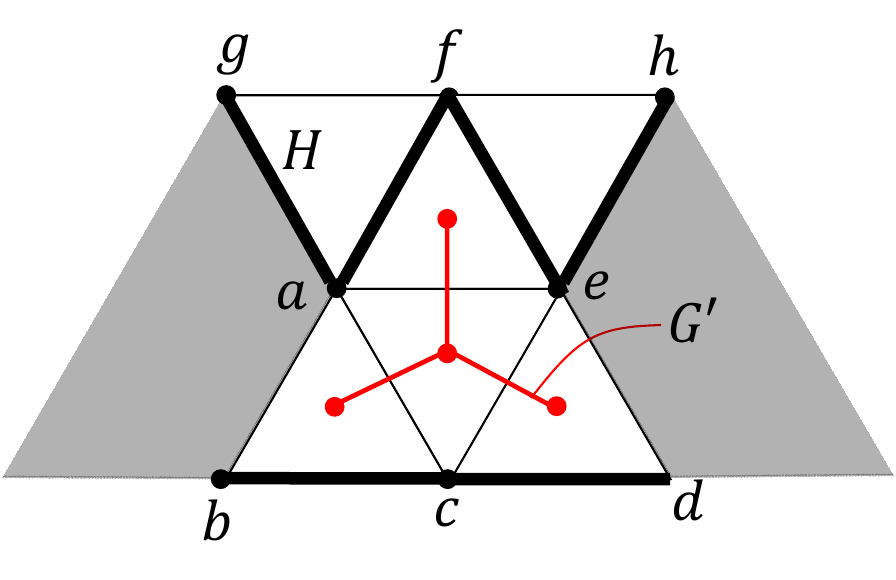}
            \caption{Case~(b) when $H$ contains only the edges $(b,c)$ and $(c,d)$.}
            \label{fig:2leaf_hamilton_2}
          \end{figure}

          Finally, consider the case when $H$ contains only the edges $(b,c), (c,d)$ on $C'$.
          Then, similarly to the previous case above, the edge $(a,b)$ is on some boundary cycle in $\B$, and  
          $H$ goes through $a$ without using edges $(a,b)$, $(a,c)$, or $(a,e)$.
          By a similar argument applied for the vertex $e$, we see that $H$ goes through the vertex $e$ without using edges $(e,a)$, $(e,c)$, or $(e,d)$.
          We denote by $f,g,h$ the vertices so that $N(a)=\{b,c,e,f,g\}$ and $N(e)=\{d,c,a,f,h\}$.
          Then $H$ contains the $4$ edges $(a,g)$, $(a,f)$, $(e,f)$, and $(e,h)$.
          In this case, the diamond induced by $\{a,c,e,f\}$ satisfies the condition~(ii) of Lemma~\ref{lem:LC_localstructure}.
          
          Thus, if $G'$ has a vertex $v$ of degree $3$ adjacent to two leaves $u$ and $w$, there exist a diamond and a Hamilton cycle satisfying one of the conditions of Lemma~\ref{lem:LC_localstructure}.

  Therefore, Lemma~\ref{lem:LC_localstructure} holds.
\end{proof}

\section{Concluding Remarks}\label{sec:conclusion}

         \begin{figure}[t]
            \centering
            \includegraphics[keepaspectratio,width=0.4\textwidth]
            {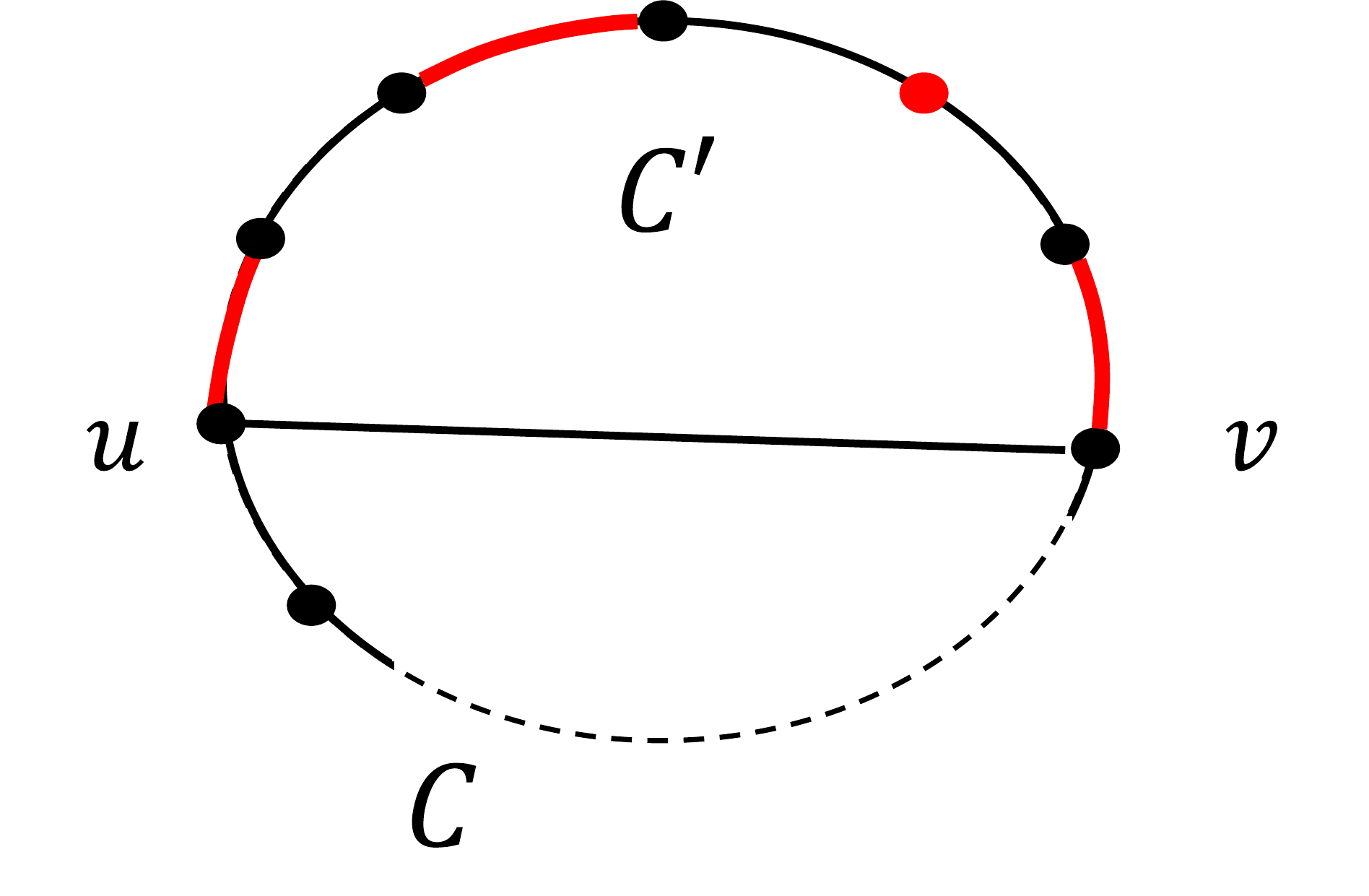}
            \caption{An odd cycle with one chord.}
            \label{fig:conclusion}
          \end{figure}

In this paper, we introduced a new reconfiguration problem of labeled matchings in a triangular grid graph.
We provided sufficient conditions for a graph to be reconfigurable using a factor-critical graphs and a locally-connected graphs.
It remains open to characterize a reconfigurable triangular grid graph, when it is a factor-critical graph with no vertex of degree $6$, but not locally-connected.
Let us here discuss the difficulty to obtain the characterization.
For example, consider a graph $G$ consisting of an odd cycle $C$ of length $2n+1$ with one chord $(u, v)$~(Figure~\ref{fig:conclusion}).
Let $C'$ be the odd cycle of $G$ with the edge $(u, v)$.
The length of $C'$ is denoted by $2m+1$.
Then we can observe that the reconfigurability of $G$ depends on $m$ and $n$.
Specifically, $G$ is reconfigurable if and only if $n-1$ and $m-1$ are mutually prime.
Indeed, since we can only rotate a placement along $C'$ and $C$, which correspond to cyclic permutations on $[m]$ and $[n]$, respectively, any permutation can be realized if and only if $n-1$ and $m-1$ are mutually prime.
This observation would imply that it requires algebraic conditions to characterize a reconfigurable graph, like the 15-puzzle.

\bibliography{main}

\begin{thebibliography}{10}

\bibitem{AichholzerDKLLM22}
Oswin Aichholzer, Erik~D. Demaine, Matias Korman, Anna Lubiw, Jayson Lynch, Zuzana Mas{\'{a}}rov{\'{a}}, Mikhail Rudoy, Virginia {Vassilevska Williams}, and Nicole Wein.
\newblock Hardness of token swapping on trees.
\newblock In Shiri Chechik, Gonzalo Navarro, Eva Rotenberg, and Grzegorz Herman, editors, {\em 30th Annual European Symposium on Algorithms, {ESA} 2022, September 5-9, 2022, Berlin/Potsdam, Germany}, volume 244 of {\em LIPIcs}, pages 3:1--3:15. Schloss Dagstuhl - Leibniz-Zentrum f{\"{u}}r Informatik, 2022.
\newblock URL: \url{https://doi.org/10.4230/LIPIcs.ESA.2022.3}, \href {https://doi.org/10.4230/LIPICS.ESA.2022.3} {\path{doi:10.4230/LIPICS.ESA.2022.3}}.

\bibitem{BartierBM23}
Valentin Bartier, Nicolas Bousquet, and Amer~E. Mouawad.
\newblock Galactic token sliding.
\newblock {\em J. Comput. Syst. Sci.}, 136:220--248, 2023.
\newblock URL: \url{https://doi.org/10.1016/j.jcss.2023.03.008}, \href {https://doi.org/10.1016/J.JCSS.2023.03.008} {\path{doi:10.1016/J.JCSS.2023.03.008}}.

\bibitem{BonamyBHIKMMW19}
Marthe Bonamy, Nicolas Bousquet, Marc Heinrich, Takehiro Ito, Yusuke Kobayashi, Arnaud Mary, Moritz M{\"{u}}hlenthaler, and Kunihiro Wasa.
\newblock The perfect matching reconfiguration problem.
\newblock In Peter Rossmanith, Pinar Heggernes, and Joost{-}Pieter Katoen, editors, {\em 44th International Symposium on Mathematical Foundations of Computer Science, {MFCS} 2019, August 26-30, 2019, Aachen, Germany}, volume 138 of {\em LIPIcs}, pages 80:1--80:14. Schloss Dagstuhl - Leibniz-Zentrum f{\"{u}}r Informatik, 2019.
\newblock URL: \url{https://doi.org/10.4230/LIPIcs.MFCS.2019.80}, \href {https://doi.org/10.4230/LIPICS.MFCS.2019.80} {\path{doi:10.4230/LIPICS.MFCS.2019.80}}.

\bibitem{BonnetMR18}
{\'{E}}douard Bonnet, Tillmann Miltzow, and Pawel Rzazewski.
\newblock Complexity of token swapping and its variants.
\newblock {\em Algorithmica}, 80(9):2656--2682, 2018.
\newblock URL: \url{https://doi.org/10.1007/s00453-017-0387-0}, \href {https://doi.org/10.1007/S00453-017-0387-0} {\path{doi:10.1007/S00453-017-0387-0}}.

\bibitem{abs-2204-10526}
Nicolas Bousquet, Amer~E. Mouawad, Naomi Nishimura, and Sebastian Siebertz.
\newblock A survey on the parameterized complexity of the independent set and (connected) dominating set reconfiguration problems.
\newblock {\em CoRR}, abs/2204.10526, 2022.
\newblock URL: \url{https://doi.org/10.48550/arXiv.2204.10526}, \href {https://arxiv.org/abs/2204.10526} {\path{arXiv:2204.10526}}, \href {https://doi.org/10.48550/ARXIV.2204.10526} {\path{doi:10.48550/ARXIV.2204.10526}}.

\bibitem{BrunnerCDHHSZ21}
Josh Brunner, Lily Chung, Erik~D. Demaine, Dylan~H. Hendrickson, Adam Hesterberg, Adam Suhl, and Avi Zeff.
\newblock 1 {X} 1 rush hour with fixed blocks is pspace-complete.
\newblock In Martin Farach{-}Colton, Giuseppe Prencipe, and Ryuhei Uehara, editors, {\em 10th International Conference on Fun with Algorithms, {FUN} 2021, May 30 to June 1, 2021, Favignana Island, Sicily, Italy}, volume 157 of {\em LIPIcs}, pages 7:1--7:14. Schloss Dagstuhl - Leibniz-Zentrum f{\"{u}}r Informatik, 2021.
\newblock URL: \url{https://doi.org/10.4230/LIPIcs.FUN.2021.7}, \href {https://doi.org/10.4230/LIPICS.FUN.2021.7} {\path{doi:10.4230/LIPICS.FUN.2021.7}}.

\bibitem{BuchinB12}
Kevin Buchin and Maike Buchin.
\newblock Rolling block mazes are pspace-complete.
\newblock {\em J. Inf. Process.}, 20(3):719--722, 2012.
\newblock URL: \url{https://doi.org/10.2197/ipsjjip.20.719}, \href {https://doi.org/10.2197/IPSJJIP.20.719} {\path{doi:10.2197/IPSJJIP.20.719}}.

\bibitem{CardinalS23}
Jean Cardinal and Raphael Steiner.
\newblock Inapproximability of shortest paths on perfect matching polytopes.
\newblock In Alberto~Del Pia and Volker Kaibel, editors, {\em Integer Programming and Combinatorial Optimization - 24th International Conference, {IPCO} 2023, Madison, WI, USA, June 21-23, 2023, Proceedings}, volume 13904 of {\em Lecture Notes in Computer Science}, pages 72--86. Springer, 2023.
\newblock \href {https://doi.org/10.1007/978-3-031-32726-1\_6} {\path{doi:10.1007/978-3-031-32726-1\_6}}.

\bibitem{DemaineDFHIOOUY15}
Erik~D. Demaine, Martin~L. Demaine, Eli Fox{-}Epstein, Duc~A. Hoang, Takehiro Ito, Hirotaka Ono, Yota Otachi, Ryuhei Uehara, and Takeshi Yamada.
\newblock Linear-time algorithm for sliding tokens on trees.
\newblock {\em Theor. Comput. Sci.}, 600:132--142, 2015.
\newblock URL: \url{https://doi.org/10.1016/j.tcs.2015.07.037}, \href {https://doi.org/10.1016/J.TCS.2015.07.037} {\path{doi:10.1016/J.TCS.2015.07.037}}.

\bibitem{DemaineR18}
Erik~D. Demaine and Mikhail Rudoy.
\newblock A simple proof that the {($n^2 -1$)}-puzzle is hard.
\newblock {\em Theor. Comput. Sci.}, 732:80--84, 2018.
\newblock URL: \url{https://doi.org/10.1016/j.tcs.2018.04.031}, \href {https://doi.org/10.1016/J.TCS.2018.04.031} {\path{doi:10.1016/J.TCS.2018.04.031}}.

\bibitem{FlakeB02}
Gary~William Flake and Eric~B. Baum.
\newblock Rush hour is pspace-complete, or "why you should generously tip parking lot attendants".
\newblock {\em Theor. Comput. Sci.}, 270(1-2):895--911, 2002.
\newblock \href {https://doi.org/10.1016/S0304-3975(01)00173-6} {\path{doi:10.1016/S0304-3975(01)00173-6}}.

\bibitem{ref:triangle}
Valery~S. Gordon, Yury~L. Orlovich, and Frank Werner.
\newblock Hamiltonian properties of triangular grid graphs.
\newblock {\em Discret. Math.}, 308(24):6166--6188, 2008.
\newblock URL: \url{https://doi.org/10.1016/j.disc.2007.11.040}, \href {https://doi.org/10.1016/J.DISC.2007.11.040} {\path{doi:10.1016/J.DISC.2007.11.040}}.

\bibitem{ref:gourds}
Joep Hamersma, Marc~J. van Kreveld, Yushi Uno, and Tom~C. van~der Zanden.
\newblock Gourds: {A} sliding-block puzzle with turning.
\newblock In Yixin Cao, Siu{-}Wing Cheng, and Minming Li, editors, {\em 31st International Symposium on Algorithms and Computation, {ISAAC} 2020, December 14-18, 2020, Hong Kong, China (Virtual Conference)}, volume 181 of {\em LIPIcs}, pages 33:1--33:16. Schloss Dagstuhl - Leibniz-Zentrum f{\"{u}}r Informatik, 2020.
\newblock URL: \url{https://doi.org/10.4230/LIPIcs.ISAAC.2020.33}, \href {https://doi.org/10.4230/LIPICS.ISAAC.2020.33} {\path{doi:10.4230/LIPICS.ISAAC.2020.33}}.

\bibitem{HearnD05}
Robert~A. Hearn and Erik~D. Demaine.
\newblock Pspace-completeness of sliding-block puzzles and other problems through the nondeterministic constraint logic model of computation.
\newblock {\em Theor. Comput. Sci.}, 343(1-2):72--96, 2005.
\newblock URL: \url{https://doi.org/10.1016/j.tcs.2005.05.008}, \href {https://doi.org/10.1016/J.TCS.2005.05.008} {\path{doi:10.1016/J.TCS.2005.05.008}}.

\bibitem{HearnDemaine}
Robert~A. Hearn and Erik~D. Demaine.
\newblock {\em Games, puzzles and computation}.
\newblock A {K} Peters, 2009.

\bibitem{ItoDHPSUU11}
Takehiro Ito, Erik~D. Demaine, Nicholas J.~A. Harvey, Christos~H. Papadimitriou, Martha Sideri, Ryuhei Uehara, and Yushi Uno.
\newblock On the complexity of reconfiguration problems.
\newblock {\em Theor. Comput. Sci.}, 412(12-14):1054--1065, 2011.
\newblock URL: \url{https://doi.org/10.1016/j.tcs.2010.12.005}, \href {https://doi.org/10.1016/J.TCS.2010.12.005} {\path{doi:10.1016/J.TCS.2010.12.005}}.

\bibitem{ItoKKKO22}
Takehiro Ito, Naonori Kakimura, Naoyuki Kamiyama, Yusuke Kobayashi, and Yoshio Okamoto.
\newblock Shortest reconfiguration of perfect matchings via alternating cycles.
\newblock {\em {SIAM} J. Discret. Math.}, 36(2):1102--1123, 2022.
\newblock URL: \url{https://doi.org/10.1137/20m1364370}, \href {https://doi.org/10.1137/20M1364370} {\path{doi:10.1137/20M1364370}}.

\bibitem{AJM1879}
Wm.~Woolsey Johnson and William~E. Story.
\newblock Notes on the "15" puzzle.
\newblock {\em American Journal of Mathematics}, 2(4):397--404, 1879.
\newblock URL: \url{http://www.jstor.org/stable/2369492}.

\bibitem{Kim16}
Dohan Kim.
\newblock Sorting on graphs by adjacent swaps using permutation groups.
\newblock {\em Comput. Sci. Rev.}, 22:89--105, 2016.
\newblock URL: \url{https://doi.org/10.1016/j.cosrev.2016.09.003}, \href {https://doi.org/10.1016/J.COSREV.2016.09.003} {\path{doi:10.1016/J.COSREV.2016.09.003}}.

\bibitem{ref:ear_decomposition}
L.~Lov{\'a}sz and M.D. Plummer.
\newblock {\em Matching Theory}.
\newblock AMS Chelsea Publishing Series. AMS Chelsea Pub., 2009.

\bibitem{MiltzowNORTU16}
Tillmann Miltzow, Lothar Narins, Yoshio Okamoto, G{\"{u}}nter Rote, Antonis Thomas, and Takeaki Uno.
\newblock Approximation and hardness of token swapping.
\newblock In Piotr Sankowski and Christos~D. Zaroliagis, editors, {\em 24th Annual European Symposium on Algorithms, {ESA} 2016, August 22-24, 2016, Aarhus, Denmark}, volume~57 of {\em LIPIcs}, pages 66:1--66:15. Schloss Dagstuhl - Leibniz-Zentrum f{\"{u}}r Informatik, 2016.
\newblock URL: \url{https://doi.org/10.4230/LIPIcs.ESA.2016.66}, \href {https://doi.org/10.4230/LIPICS.ESA.2016.66} {\path{doi:10.4230/LIPICS.ESA.2016.66}}.

\bibitem{N18}
Naomi Nishimura.
\newblock Introduction to reconfiguration.
\newblock {\em Algorithms}, 11(4):52, 2018.
\newblock \href {https://doi.org/10.3390/a11040052} {\path{doi:10.3390/a11040052}}.

\bibitem{RatnerW86}
Daniel Ratner and Manfred~K. Warmuth.
\newblock Finding a shortest solution for the {N} {\texttimes} {N} extension of the 15-puzzle is intractable.
\newblock In Tom Kehler, editor, {\em Proceedings of the 5th National Conference on Artificial Intelligence. Philadelphia, PA, USA, August 11-15, 1986. Volume 1: Science}, pages 168--172. Morgan Kaufmann, 1986.
\newblock URL: \url{http://www.aaai.org/Library/AAAI/1986/aaai86-027.php}.

\bibitem{Sanita18}
Laura Sanit{\`{a}}.
\newblock The diameter of the fractional matching polytope and its hardness implications.
\newblock In {\em 59th {IEEE} Annual Symposium on Foundations of Computer Science, {FOCS} 2018, Paris, France, October 7-9, 2018}, pages 910--921, 2018.
\newblock \href {https://doi.org/10.1109/FOCS.2018.00090} {\path{doi:10.1109/FOCS.2018.00090}}.

\bibitem{schrijver-book}
Alexander Schrijver.
\newblock {\em Combinatorial Optimization: Polyhedra and Efficiency}.
\newblock Springer, Berlin, 2003.

\bibitem{slocum200615}
J.~Slocum and D.~Sonneveld.
\newblock {\em The 15 Puzzle Book: How It Drove the World Crazy}.
\newblock Slocum Puzzle Foundation, 2006.

\bibitem{Heuvel13}
Jan van~den Heuvel.
\newblock The complexity of change.
\newblock In Simon~R. Blackburn, Stefanie Gerke, and Mark Wildon, editors, {\em Surveys in Combinatorics 2013}, volume 409 of {\em London Mathematical Society Lecture Note Series}, pages 127--160. Cambridge University Press, 2013.
\newblock \href {https://doi.org/10.1017/CBO9781139506748.005} {\path{doi:10.1017/CBO9781139506748.005}}.

\end{thebibliography}

\end{document}